\documentclass[12pt]{interact}
\usepackage{soul}
\usepackage{natbib}

\newcommand{\tilnu}{\tilde{\nu}}
\newcommand{\aaaa}{\mathrm{(a)}}
\newcommand{\bbbb}{\mathrm{(b)}}
\newcommand{\cccc}{\mathrm{(c)}}
\newcommand{\dddd}{\mathrm{(d)}}
\newcommand{\eeee}{\mathrm{(e)}}

\newcommand{\lp}{\left(}
\newcommand{\rp}{\right)}
\newcommand{\lb}{\left[}
\newcommand{\rb}{\right]}
\newcommand{\lbp}{\left\{}
\newcommand{\rbp}{\right\}}
\newcommand{\lba}{\left\lvert}
\newcommand{\rba}{\right\rvert}

\newcommand{\mcal}{\mathcal}
\newcommand{\mrm}{\mathrm}

\newcommand{\mbb}{\mathbb}

\newcommand{\lce}{\left\lceil}
\newcommand{\rce}{\right\rceil}
\newcommand{\lfl}{\left\lfloor}
\newcommand{\rfl}{\right\rfloor}

\newcommand{\diid}{\overset{\text{i.i.d.}}{\sim}}

\newcommand{\E}{\mathbb{E}}

\renewcommand{\Pr}{\mathbb{P}}

\newcommand{\argmin}{\mathop{\mathrm{argmin}}}

\makeatletter
\newcommand*{\indep}{%
  \mathbin{%
    \mathpalette{\@indep}{}%
  }%
}
\newcommand*{\nindep}{%
  \mathbin{
    \mathpalette{\@indep}{\not}
  }%
}
\newcommand*{\@indep}[2]{%
  \sbox0{$#1\perp\m@th$}
  \sbox2{$#1=$}
  \sbox4{$#1\vcenter{}$}
  \rlap{\copy0}
  \dimen@=\dimexpr\ht2-\ht4-.2pt\relax
  \kern\dimen@
  {#2}%
  \kern\dimen@
  \copy0 
} 
\makeatother

\theoremstyle{plain}
\newtheorem{theorem}{Theorem}
\newtheorem{corollary}{Corollary}
\newtheorem{lemma}{Lemma}
\newtheorem{remark}{Remark}

\renewcommand\abstractname{Abstract}

\begin{document}

\allowdisplaybreaks

\onecolumn 

\title{Finite-Horizon  Quickest Change Detection Balancing Latency with False Alarm Probability}

\titleheader{Finite-Horizon QCD with FA Probability-Latency Tradeoff}

\author{
\name{Yu-Han Huang\thanks{CONTACT Y.-H. Huang Email: yuhanhh2@illinois.edu} and Venugopal V. Veeravalli\thanks{CONTACT V. V. Veeravalli Email: vvv@illinois.edu}}
\affil{ECE and CSL, The Grainger College of Engineering, University of Illinois Urbana-Champaign, Urbana, IL, USA}
}

\maketitle
\begin{abstract}



A finite-horizon variant of the quickest change detection (QCD) problem that is of relevance to learning in non-stationary environments is studied. The metric characterizing false alarms is the probability of a false alarm occurring before the horizon ends. The metric that characterizes the delay is \emph{latency}, which is the smallest value such that the probability that detection delay exceeds this value is upper bounded to a predetermined latency level.
The objective is to minimize the latency (at a given latency level), while maintaining a low false alarm probability. Under  the pre-specified latency and false alarm levels, a universal lower bound on the latency, which any change detection procedure needs to satisfy, is derived. Change detectors are then developed, which are order-optimal in terms of the horizon. 
The case where the pre- and post-change distributions are known is considered first, and then the results are generalized to the non-parametric  case when they are unknown except that they are  sub-Gaussian with different means.  Simulations are provided to validate the theoretical results.


\end{abstract}


\section{Introduction}
\label{sec:Intro}


In the quickest change detection (QCD) problem, the environment undergoes a change at an unknown \emph{change-point}, which affects the  distribution of stochastic observations of the environment. 
The goal for the agent is to detect the change as soon as possible, while controlling the false alarm rate. 
QCD problems have been extensively studied due to the broad range of applications in engineering and science, including but not limited to biomedical signal processing, epicenter monitoring, financial market, and quality control in manufacturing. See \citep{poor-hadj-qcd-book-2009, vvv_qcd_overview, tart-niki-bass-2014, xie_vvv_qcd_overview} for articles and books on this topic.

QCD problems are mathematically formulated as constrained optimization problems, in which the objective is to minimize an appropriate metric of the detection delay, under a constraint on a suitable false alarm metric. 
The most common false alarm metric is the \emph{mean time to false alarm} (or equivalently, its reciprocal, the \emph{false alarm rate}). 
Detection delay is generally measured by considering the \emph{worst-case expected detection delay} (conditioned on not having a false alarm and possibly the history of samples prior to the change-point) over all possible change-points across an infinite horizon.
These metrics, however, may not be appropriate in many practical settings. First, the horizon may be large but finite in practice, which does not fit the infinite horizon formulations common in QCD literature. 
Additionally, instead of minimizing the worst-case expected delay while suppressing the false alarm rate, one might be interested in controlling the probability of false alarm over the horizon, and that of detection delay being too large, as these two ``bad" events can possibly have a significant impact on the system. 
Lai \cite[Section II.B]{lai1998information}  addresses the false alarm aspect of such a formulation partially, by considering the goal of minimizing the worst-case expected delay while controlling the false alarm probability over a window to a prespecified level. However, the size of the window does not cover the entire horizon, and it depends on the prespecified false alarm level. 
Thus, the QCD problem in which we are interested, where we balance the probability of false alarm over the horizon with \emph{latency}, which the smallest value such that the probability the detection delay exceeds this value is upper bounded to a predetermined level,  remains largely open. A key application of this formulation is in online learning in a piecewise stationary environment, which we describe next.

In a piecewise stationary environment, the underlying distribution of the environment undergoes changes at multiple change-points, and remains stationary in between. For the case of piecewise stationary bandits and piecewise stationary finite-horizon episodic Markov decision processes, there is a rich literature on algorithms that restart as soon as they detect a change in the environment \citep{auer2019adswitch, wei2021master, besson2022efficient, yadkori2023nsb, gerogiannis2024blackboxfeas, huang2025change}. To achieve good (order-optimal) performance, these algorithms need to control the probability of false alarm and that of detection delay being too large \citep{besson2022efficient, huang2025change}. 

In this work, we formulate a QCD problem in which the metrics are the false alarm probability and latency. 
The objective is to minimize the latency (at a prespecified latency level), while maintaining the probability of false alarm below a prespecified level. Under the prespecified latency and false alarm levels, we first establish a universal lower bound on the latency that any change detector must satisfy. Then, we propose change detectors that are order-optimal in terms of the horizon (as the horizon goes to infinity) in scenarios with full knowledge of the underlying distributions. We then generalize these order-optimal change detectors to scenarios where the underlying distributions are unknown, except for being sub-Gaussian with different means.  Preliminary versions of these results are described in the following conference proceedings \citep{huang2024high, huang2025sequentialchangedetectionlearning}.

The remainder of the paper is organized as follows: The QCD problem of interest and the metrics for detection delay (latency) and false alarm are presented in Section \ref{sec:ProbForm}. A universal lower bound on latency is derived in Section \ref{sec:low_bound}. Order-optimal change detectors are discussed in Section \ref{sec:CD}. The performance of proposed change detectors is validated through simulations in Section \ref{sec:sim}. Lastly, some concluding remarks, including avenues for future research  are presented in Section \ref{sec:Sum}.

\section{Problem Formulation}
\label{sec:ProbForm}

For any positive integers $k \leq n$, we denote the set $\lbp k, \dots, n \rbp$ by $[k, n]$, and denote $[1, n]$ by $[n]$ for conciseness. Let $\lbp X_{n}: n \in [T] \rbp$ be a set of independent random vectors observed by an agent over a finite horizon $T$ sequentially. At an unknown time-step, referred to as the \emph{change-point} $\nu$, the distribution of the observations undergoes a change, i.e.,
\begin{equation} \label{eq:sample_distr}
    X_{n} \sim \begin{dcases}
        f_{0}, & n \in [\nu - 1] \\
        f_{1}, & n \in [\nu, T]
    \end{dcases}.
\end{equation}
This indicates that the observations before the change-point $\nu$ follow the pre-change density $f_{0}$, while those after the change-point follow the post-change density $f_{1}$. Both densities are defined with respect to a common dominating measure $\lambda$. 
When $f_{0}$ is unknown, the agent needs a small number of pre-change observations to estimate $f_{0}$. Thus, following \citep{lai2010sequential}, we assume that the change-point does not occur within the first $m$ time-steps, i.e., $\nu > m$. These first $m$ time-steps, known as the \emph{pre-change window}, are necessary for change detectors to learn the pre-change density.

Let $\tau$ denote the stopping time of a (causal) change detector, i.e., the time-step at which a change detector declares a change. Following notation common in the QCD literature, we use $\Pr_{\nu}$ and $\E_{\nu}$ to denote the probability measure and the expectation when the change occurs at $\nu$. Similarly, $\Pr_{\infty}$ and $\E_{\infty}$ represent the probability measure and the expectation when there is no change in the distribution. As discussed in Section \ref{sec:Intro}, our metric for detection delay is the smallest value  such that this value is exceeded by the detection delay $\tau - \nu$ with low probability, for any possible change-point. Therefore, we define the latency as follows: for some (small) latency level $\delta_{\mrm{D}} \in (0,1)$,
\begin{align} \label{eq:latency}
    \ell \coloneqq \min \lbp d \in [T]: \Pr_{\nu} \lp \tau \geq \nu + d \rp \leq \delta_{\mathrm{D}}, \forall\, \nu \in \lb m+1, T-d \rb  \rbp
\end{align}
where $m = 0$ when the pre-change density $f_{0}$ is known. Then, the goal of a change detector is to minimize the latency, subject to the constraint that the false alarm probability over a finite horizon $\Pr_{\infty} \lp \tau \leq T \rp$ remains below a (small) false alarm level $\delta_{\mrm{F}} \in (0,1)$. Therefore, we propose the following formulation of the QCD problem of interest: for false alarm and latency levels $\delta_{\mrm{F}}, \delta_{\mrm{D}} \in \lp 0, 1 \rp$,
\begin{equation}
    \underset{\tau}{\textrm{minimize}} \enspace \ell \quad\quad \textrm{s.t.}\enspace\Pr_{\infty}\lp\tau\leq T\rp\leq\delta_{\mathrm{F}}
\label{eq:QCD}
\end{equation}
where $m = 0$ when the pre-change density $f_{0}$ is known. Note that when we develop candidate change detector for this problem, we do \emph{not} leverage the knowledge of the horizon $T$, as this information about the horizon might not be available in general settings.
%

%
%

\section{Universal Lower Bound on Latency} \label{sec:low_bound}
In order to design good change detectors for the problem introduced in \eqref{eq:QCD}, we first establish the theoretical lower bound on the latency that any change detector must satisfy, with the understanding that good change detectors should approach this lower bound in terms of their latency for prespecified false alarm and latency levels $(\delta_{\mrm{F}}, \delta_{\mrm{D}})$. 
%
The following quantity appears in the lower bound:
\begin{align}
K\coloneqq\log\lp\E_{f_{1}}\lb\frac{f_{1}\lp X\rp}{f_{0}\lp X\rp}\rb\rp \label{eq:C_const}   
\end{align}
where $\E_{f_{1}}$ denotes the expectation with $X \sim f_{1}$. Using Jensen's inequality, it is easily seen that $K>D(f_{1}||f_{0})>0$ where $D$ denotes Kullback-Liebler (KL) divergence.  We further assume that $K<\infty$.

\begin{theorem}[Asymptotic Lower Bound on Latency] \label{thm:low_bound_latency} For all $\delta_{\mathrm{F}}, \delta_{\mathrm{D}} \in \lp 0, 1 \rp$ such that $\delta_{\mathrm{F}} + \delta_{\mathrm{D}}<1$
\begin{equation}
\ell \geq \lp\frac{1}{K} + o\lp 1\rp\rp \lb\log\lp T\rp + \log\lp \frac{1}{\delta_{\mathrm{F}}}\rp + \log\lp1 - \delta_{\mathrm{F}} - \delta_{\mathrm{D}}\rp + o\lp 1\rp\rb\label{eq:d_lower}
\end{equation}
as $T \to \infty$.
\end{theorem}
The proof of Theorem \ref{thm:low_bound_latency} is given in Appendix \ref{sec:thm1}. 

\section{Candidate Change Detectors}
\label{sec:CD}

Given the universal lower bound of Theorem \ref{thm:low_bound_latency}, our goal in this section is to design change detectors that approach this bound asymptotically as $T \to \infty$. We first consider the scenario with full knowledge of the underlying distributions, and then generalize to the scenario where the
underlying distributions are unknown, except for being sub-Gaussian with different
means.


\subsection{Known Pre- and Post-Change Distributions}
\label{sec:pre-post-known}

In the setting where $T=\infty$ and the independent stochastic observations follow \eqref{eq:sample_distr}, it was demonstrated by Moustakides \citep{moustakides1986optimal} that Page's Cumulative Sum (CuSum) test \citep{page1954continuous} is exactly optimal under Lorden's criterion \citep{lorden1971procedures}. The CuSum test employs the CuSum statistic:
\begin{equation}\label{eq:CuSum_stat}
    C_{n} \coloneqq \sup_{k \in \lb n\rb} \sum_{i = k}^{n} \log \lp \frac{f_{1} \lp X_{i}\rp}{f_{0} \lp X_{i}\rp} \rp,\; n \in \mbb{N}.
\end{equation}
which satisfies the recursion:
\begin{equation}\label{eq:CuSum_recur}
    C_{n} = \max\lbp C_{n-1}, 0\rbp + \log \lp \frac{f_{1} \lp X_{n}\rp}{f_{0} \lp X_{n}\rp} \rp,\; n \in \mbb{N},
\end{equation}
with $C_{0} \coloneqq 0$. The CuSum test declares a change when the CuSum statistic exceeds a constant threshold. Thus, the stopping time of the CuSum test is defined as follows:
\begin{equation}\label{eq:CuSum_tau}
    \tau_{\mrm{CS}} \coloneqq \inf\lbp n \in \mbb{N}: C_{n} \geq b \rbp,
\end{equation}
where $b$ is the threshold. In the same setting where $T = \infty$ and observations follow \eqref{eq:sample_distr}, the Shiryaev–Roberts (SR) test  \citep{shiryaev1961problem} is shown to be asymptotically optimal under  Pollak's criterion \citep{pollak1985optimal} and under Lorden's criterion. The SR test utilizes the statistic:
\begin{equation}\label{eq:SR_stat}
    S_{n} \coloneqq \sum_{k \in \lb n\rb} \prod_{i = k}^{n} \frac{f_{1} \lp X_{i}\rp}{f_{0} \lp X_{i}\rp},\; n \in \mbb{N},
\end{equation}
which satisfies the recursion:
\begin{align}\label{eq:SR_recur}
    S_{n} = \lp S_{n-1} + 1\rp \frac{f_{1} \lp X_{n}\rp}{f_{0} \lp X_{n}\rp},\; n \in \mbb{N},
\end{align}
where $S_{0} \coloneqq 0$. Similar to the CuSum test, the SR test gets triggered when the SR statistic surpasses a constant threshold. The stopping time of the SR test can then be defined as follows:
\begin{equation}\label{eq:SR_tau}
    \tau_{\mrm{SR}} \coloneqq \inf\lbp n \in \mbb{N}: S_{n} \geq b \rbp.
\end{equation}
Due to the recursive properties in \eqref{eq:CuSum_recur} and \eqref{eq:SR_recur}, the CuSum and SR tests are computationally efficient, as the computational complexity of the statistics per time-step is $\mcal{O}(1)$.
Owing to the aforementioned optimality of the CuSum and SR tests, it is natural to speculate that these tests can attain the lower bound in Theorem \ref{thm:low_bound_latency} asymptotically as $T \to \infty$. 
However, since the mean time to false alarm is finite for the CuSum and SR tests (i.e., $\E_{\infty} \lb \tau_{\mrm{CS}} \rb, \E_{\infty} \lb \tau_{\mrm{SR}} \rb < \infty$), the false alarm probability of these tests goes to $1$ as the horizon goes to infinity (i.e., $\lim_{T \to \infty} \Pr_{\infty}\lp \tau_{\mrm{CS}} \leq T\rp = \lim_{T \to \infty} \Pr_{\infty}\lp \tau_{\mrm{SR}} \leq T\rp = 1$). This indicates that these tests (with a constant threshold) cannot suppress the false alarm probability over a large horizon if the threshold cannot be tuned according to the horizon. As the agent is assumed to be oblivious to the horizon, the CuSum test in \eqref{eq:CuSum_tau} and the SR test in \eqref{eq:SR_tau} are not good candidate solutions. 

One way to control the false alarm rate in the CuSum and SR tests for large horizons is to increase the thresholds of these tests with time. Taking cues from the analysis in  \citep{kaufmann2021mixture, besson2022efficient}, we propose the Time-Varying-Threshold CuSum (TVT-CuSum) and Time-Varying-Threshold SR (TVT-SR) tests, whose thresholds increase logarithmically with time. The stopping time of the TVT-CuSum test is defined as follows: for $r > 1$,
\begin{align}\label{eq:TVT-CuSum_tau}
    \tau_{\mrm{C}, r} \coloneqq \inf \lbp n \in \mbb{N}: C_{n} \geq \beta_{\mrm{C}} \lp n, \delta_{\mrm{F}}, r \rp \rbp,
\end{align}
where the time-varying threshold $\beta_{\mrm{C}}$ is:
\begin{align}\label{eq:TVT-CuSum_thres}
    \beta_{\mrm{C}} \lp n, \delta_{\mrm{F}}, r \rp \coloneqq \log\lp \zeta \lp r\rp \frac{n^{r}}{\delta_{\mrm{F}}} \rp,\; n \in \mbb{N},
\end{align}
with $\zeta \lp r \rp = \sum_{i = 1}^{\infty} i^{-r}$. Similarly, the stopping time of the TVT-SR test is defined as follows: for $r > 1$,
\begin{align}\label{eq:TVT-SR_tau}
    \tau_{\mrm{S}, r} \coloneqq \inf \lbp n \in \mbb{N}: \log S_{n} \geq \beta_{\mrm{S}} \lp n, \delta_{\mrm{F}}, r \rp \rbp,
\end{align}
where the time-varying threshold $\beta_{\mrm{S}}$ has the following form:
\begin{align}\label{eq:TVT-SR_thres}
    \beta_{\mrm{S}} \lp n, \delta_{\mrm{F}}, r \rp \coloneqq \beta_{\mrm{C}} \lp n, \delta_{\mrm{F}}, r \rp + \log n,\; n \in \mbb{N}.
\end{align}
In the following theorem, we establish that the TVT-CuSum and TVT-SR tests control the false alarm probability to be below the false alarm level $\delta_{\mrm{F}}$. Furthermore, the latencies of these tests at latency level $\delta_{\mrm{D}}$ are $\mcal{O}\lp \log T \rp$, matching the order in the lower bound in Theorem \ref{thm:low_bound_latency}. In the theorem, we use $\Lambda$ to denote the cumulant generating function of $\log \frac{f_{0} \lp X \rp}{f_{1} \lp X \rp}$ with $X \sim f_{1}$, i.e.,
%
\begin{equation} \label{eq:cumulant}
\Lambda\lp\theta\rp=\log\lp\E_{f_{1}}\lb\exp\lp\theta\log\lp\frac{f_{0}\lp X\rp}{f_{1}\lp X\rp}\rp\rp\rb\rp.
\end{equation}
\begin{theorem} [TVT-CuSum and TVT-SR tests] \label{thm:latency_pre-post-known}
Consider the TVT-CuSum test in \eqref{eq:TVT-CuSum_tau} and the TVT-SR test in \eqref{eq:TVT-SR_tau}. The false alarm probability of these tests is smaller than $\delta_{\mrm{F}}$, i.e., for any $r>1$,
\begin{equation}\label{FA_pre-post-known}
    \Pr_{\infty} \lp \tau_{\mrm{C},r} \leq T \rp \leq \delta_{\mrm{F}} \enspace\mrm{and}\enspace \Pr_{\infty} \lp \tau_{\mrm{S},r} \leq T \rp \leq \delta_{\mrm{F}}.
\end{equation}
In addition, the latencies of these tests are upper bounded as follows:
\begin{align}\label{eq:latency_pre-post-known}
    \ell \leq \inf_{\theta\in\lp0,1\rp} \lbp \frac{1}{\lvert\Lambda\lp\theta\rp\rvert} \lb\log\lp\frac{1}{\delta_{\mathrm{D}}}\rp +\theta\beta\lp T, \delta_{\mrm{F}}, r\rp\rb\rbp,
\end{align}
with $\beta = \beta_{\mrm{C}}$ for the TVT-CuSum test and $\beta = \beta_{\mrm{S}}$ for the TVT-SR test.
\end{theorem}

The proof of Theorem \ref{thm:latency_pre-post-known} is presented in Appendix \ref{sec:thm2}. Note that $\beta_{\mrm{C}} \lp T, \delta_{\mrm{F}}, r\rp$ and $\beta_{\mrm{S}} \lp T, \delta_{\mrm{F}}, r\rp$ are $\mcal{O} \lp \log T \rp$, matching the order of the lower bound in Theorem \ref{thm:low_bound_latency}. Hence, Theorem \ref{thm:latency_pre-post-known} indicates that the TVT-CuSum and TVT-SR tests are order-optimal with respect to the horizon $T$. 
\begin{remark}[Decreasing false alarm and latency levels] In the setting where $\delta_{\mrm{F}} = \mcal{O} \lp T^{-p} \rp$ and $\delta_{\mrm{D}} = \mcal{O} \lp T^{-q} \rp$ for some $p,q>0$, the lower bound in Theorem \ref{thm:low_bound_latency} is $\Omega \lp \log T \rp$, and the upper bound on latencies in \eqref{eq:latency_pre-post-known} is also $\mcal{O} \lp \log T \rp$. This shows that the TVT-CuSum and TVT-SR tests are also order-optimal with respect to $T$ in this setting, which is relevant to learning in piecewise stationary environments -- see, also \citep{huang2025sequentialchangedetectionlearning}.
\end{remark}

\subsection{Unknown Pre- and Post-Change Distributions}
\label{sec:pre-post-unknown}

While the TVT-CuSum and TVT-SR tests described in Section~\ref{sec:pre-post-known}  are order-optimal with respect to $T$, the computation of the CuSum and SR statistics requires full knowledge of $f_{0}$ and $f_{1}$.
%
We now investigate the setting where $f_{0}$ and $f_{1}$ are both $\sigma^{2}$-sub-Gaussian, and that the agent knows the sub-Gaussianity parameter $\sigma^{2}$.
We stress that the assumption is not too restrictive, since bounded stochastic observations (which are commonly encountered in practice) are $\sigma^{2}$-sub-Gaussian, where the parameter $\sigma$ can be derived from the support of the observations.
Let $\mu_{0}$ and $\mu_{1}$ represent the pre- and post-change means, respectively. Then, we define the change-gap $\Delta$ as the absolute difference between the pre- and post-change means, i.e., $\Delta \coloneqq \lba \mu_{0} - \mu_{1} \rba$. We assume that the change-gap is nonzero i.e., $\Delta>0$.

To enable the TVT-CuSum and TVT-SR tests to operate without knowledge of the pre- and post-change distributions, we replace the CuSum and SR statistics with their generalized versions---a common generalization method for settings with unknown pre- and post-change distributions \citep{lai1998information, lai2010sequential}. 
In particular, we replace the CuSum statistic in \eqref{eq:TVT-CuSum_tau} with the Generalized Likelihood Ratio (GLR) statistic, defined as follows:
%
%
\begin{equation} \label{eq:GLR_stat_pre-post-unknown}
    G_{n} \coloneqq \sup_{k \in \lb n \rb} \log  \frac{\sup_{\mu_{0}' \in \mbb{R}} \sup_{\mu_{1}' \in \mbb{R}} \prod_{i = 1}^{k} f_{\mu_{0}'}\lp X_{i} \rp \prod_{i = k + 1}^{n} f_{\mu_{1}'}\lp X_{i} \rp}{\sup_{\mu\in\mbb{R}} \prod_{i = 1}^{n} f_{\mu} \lp X_{i} \rp},\; n \in \lb T \rb,
\end{equation}
where $f_{\mu}$ denotes the Gaussian density with mean $\mu$ and variance $\sigma^{2}$. The following lemma (whose proof is given in the Appendix) shows that the GLR statistic can be written in terms of empirical means of the stochastic observations:

\begin{lemma}\label{lem:GLR-kl-unknown-pre-post}
For any $n\in\mbb{N}$ and any $k\in\lb n\rb$, we have:
\begin{align}
    &\log\lp\frac{\sup_{\mu_{0}'\in\mbb{R}}\prod_{i=1}^{k}f_{\mu_{0}'}\lp X_{i}\rp\sup_{\mu_{1}'\in\mbb{R}}\prod_{i=k+1}^{n}f_{\mu_{1}'}\lp X_{i}\rp}{\sup_{\mu\in\mbb{R}}\prod_{i=1}^{n}f_{\mu}\lp X_{i}\rp}\rp\nonumber\\
    &=k D \lp\hat{\mu}_{1:k};\hat{\mu}_{1:n}\rp+\lp n-k\rp D \lp\hat{\mu}_{k+1:n};\hat{\mu}_{1:n}\rp\label{eq:lr-kl}
\end{align}
where $\hat{\mu}_{m:n}$ is the empirical mean of the stochastic observations $\lbp X_{m}, \dots, X_{n} \rbp$ and \textup{$ D \lp x;y\rp\coloneqq(x-y)^{2}/(2\sigma^{2})$} is the KL-divergence between two Gaussian distributions with common variance $\sigma^{2}$ and different means $x,y\in\mbb{R}$.
\end{lemma}

Thus, for any $n \in [T]$,
\begin{align}
    G_{n} = \sup_{k\in[n]} k D \lp\hat{\mu}_{1:k};\hat{\mu}_{1:n}\rp+\lp n-k\rp D \lp\hat{\mu}_{k+1:n};\hat{\mu}_{1:n}\rp.
\end{align}
%
The corresponding GLR test has stopping time:
\begin{align}\label{eq:GLR_tau-pre-post-unknown}
    \tau_{\mrm{GLR}} \coloneqq \inf \lbp n \in \mbb{N}:\; G_{n} \geq \beta_{\mrm{GLR}} \lp n, \delta_{\mathrm{F}} \rp \rbp,
\end{align}
where the threshold function $\beta_{\mrm{GLR}}$ is defined as:  
\begin{align}\label{eq:GLR_thres-pre-post-unknown}
    \beta_{\mrm{GLR}} \lp n, \delta_{\mathrm{F}} \rp \coloneqq& 6\log\lp1+\log\lp n\rp\rp+\frac{5}{2}\log\lp\frac{4n^{3/2}}{\delta_{\mathrm{F}}}\rp + 11,\; n \in \lb T \rb.
\end{align}
Similarly, we can extend the SR statistic to the Generalized Shiryaev–Roberts (GSR) statistic to the setting where the pre- and post-change densities are unknown. The GSR statistic is defined as follows:
\begin{equation}\label{eq:GSR_stat_pre-post-unknown}
W_{n} \coloneqq \sum_{k = 1}^{n} \frac{\sup_{\mu_{0}' \in \mbb{R}} \sup_{\mu_{1}' \in \mbb{R}} \prod_{i = 1}^{k} f_{\mu_{0}'} \lp X_{i} \rp \prod_{i = k + 1}^{n} f_{\mu_{1}'}\lp X_{i}\rp}{\sup_{\mu \in \mbb{R}} \prod_{i = 1}^{n} f_{\mu} \lp X_{i} \rp},\; n \in \lb T \rb.
\end{equation}
By Lemma \ref{lem:GLR-kl-unknown-pre-post}, we can write
\begin{align}
    W_{n} = \sum_{k = 1}^{n} \exp\lp k D \lp\hat{\mu}_{1:k};\hat{\mu}_{1:n}\rp+\lp n-k\rp D \lp\hat{\mu}_{k+1:n};\hat{\mu}_{1:n}\rp \rp.
\end{align}
The stopping time of the GSR test is given by:
\begin{align}\label{eq:GSR_tau-pre-post-unknown}
    \tau_{\mrm{GSR}} \coloneqq \inf \lbp n \in \mbb{N}:\; \log W_{n} \geq \beta_{\mrm{GSR}} \lp n, \delta_{\mathrm{F}} \rp \rbp,
\end{align}
where
\begin{align}\label{eq:GSR_thres-pre-post-unknown}
    \beta_{\mrm{GSR}} \lp n, \delta_{\mathrm{F}} \rp \coloneqq& \beta_{\mrm{GLR}} \lp n, \delta_{\mathrm{F}} \rp + \log n,\; n \in \lb T \rb. 
\end{align}

In the following theorem, we show that the GLR and GSR tests can effectively control the false alarm probability at level $\delta_{\mathrm{F}}$. Furthermore, we provide upper bounds on the latencies of the tests at latency level $\delta_{\mathrm{D}}$.
\begin{theorem} [GLR and GSR tests] \label{thm:latency_pre-post-unknown}
Consider the GLR test in \eqref{eq:GLR_tau-pre-post-unknown} and the GSR test in \eqref{eq:GSR_tau-pre-post-unknown}. The false alarm probabilities of these tests are controlled at level $\delta_{\mrm{F}}$, i.e.,
\begin{equation}\label{FA_pre-post-unknown}
    \Pr_{\infty} \lp \tau_{\mrm{GLR}} \leq T \rp \leq \delta_{\mrm{F}} \enspace\mrm{and}\enspace \Pr_{\infty} \lp \tau_{\mrm{GSR}} \leq T \rp \leq \delta_{\mrm{F}}.
\end{equation}
In addition, with the pre-change window length
\begin{equation} \label{eq:m}
m \geq \frac{8\sigma^{2}}{\Delta^{2}}\beta \lp T,\delta_{\mrm{F}}\rp,
\end{equation}
the latencies of these tests at latency level $\delta_{\mathrm{D}}$ are upper bounded as follows:
\begin{equation} \label{eq:d}
  \ell\leq\lce\max\lbp\frac{8\sigma^{2}m\beta \lp T,\delta_{\mrm{F}}\rp}{\Delta^{2}m-8\sigma^{2}\beta \lp T,\delta_{\mrm{F}}\rp},\frac{\delta_{\mrm{F}}^{2/3}}{2^{16/15}\delta_{\mrm{D}}^{4/15}}-m\rbp\rce
\end{equation}
with $\beta = \beta_{\mrm{GLR}}$ in \eqref{eq:GLR_thres-pre-post-unknown}
for the GLR test and $\beta = \beta_{\mrm{GSR}}$ in \eqref{eq:GSR_thres-pre-post-unknown} for the GSR test.
\end{theorem}

The proof of Theorem \ref{thm:latency_pre-post-unknown} is given in Appendix \ref{sec:thm3}. Note that $\beta_{\mrm{GLR}} \lp T,\delta_{\mrm{F}}\rp$ and $\beta_{\mrm{GSR}} \lp T,\delta_{\mrm{F}}\rp$ are $\mcal{O} \lp \log T \rp$. In the following corollary, we show that the GLR and GSR tests are order-optimal in $T$ with 
an appropriate choice of $m$ when $\delta_{\mrm{F}}$ and $\delta_{\mrm{D}}$ are fixed.
\begin{corollary} \label{cor:m.d}
If the false alarm level $\delta_{\mrm{F}}$ and the latency level $\delta_{\mrm{D}}$ are fixed, the GLR and GSR tests have latencies (at latency level $\delta_{\mrm{D}}$) that are $\mcal{O} \lp \log T \rp$ if
\begin{equation}
    m = \lce \frac{16\sigma^{2}}{\Delta^{2}}\beta \lp T,\delta_{\mrm{F}}\rp\rce,
\end{equation}
where $\beta = \beta_{\mrm{GLR}}$ in \eqref{eq:GLR_thres-pre-post-unknown}
for the GLR test and $\beta = \beta_{\mrm{GSR}}$ in \eqref{eq:GSR_thres-pre-post-unknown} for the GSR test.
\end{corollary}

\begin{proof}
With the choice of $m$ and fixed $\delta_{\mrm{F}},\delta_{\mrm{D}} \in (0,1)$, we have
\begin{align}
\ell &= \lce\max\lbp\frac{8\sigma^{2} m\beta\lp T,\delta_{\mrm{F}}\rp}{\Delta^{2}m-8\sigma^{2}\beta\lp T,\delta_{\mrm{F}}\rp}, \frac{\delta_{\mrm{F}}^{2/3}}{2^{16/15}\delta_{\mrm{D}}^{4/15}} - m\rbp\rce \nonumber\\
&\leq \lce\frac{8\sigma^{2} m\beta\lp T,\delta_{\mrm{F}}\rp}{\Delta^{2} \lce \frac{16\sigma^{2}}{\Delta^{2}}\beta \lp T,\delta_{\mrm{F}}\rp \rce - 8\sigma^{2}\beta\lp T,\delta_{\mrm{F}}\rp}\rce \nonumber\\
&\overset{(a)}{\leq} \lce \frac{16\sigma^{2}}{\Delta^{2}}\beta \lp T,\delta_{\mrm{F}}\rp \rce \nonumber \\
&\overset{(b)}{=} \mcal{O} \lp \log T \rp,
\label{eq:d_upp_glr}
\end{align}
where step $(a)$ follows from the choice of $m$, and step $(b)$ stems from the fact that $\beta \lp T, \delta_{\mrm{F}} \rp$ is $\mcal{O} \lp \log T \rp$.
\end{proof}

\begin{remark}
The results in Theorem \ref{thm:latency_pre-post-unknown} and Corollary \ref{cor:m.d} can be generalized to the scenarios in which the stochastic observations are distributed as follows:
\begin{align} \label{eq:sample_distr_general}
    X_{n} \sim \begin{dcases}
        \mu_{0} + \eta_{n},& n \in \lb \nu-1 \rb\\
        \mu_{1} + \eta_{n},& n \in \lb \nu, T \rb
    \end{dcases}
\end{align}
where the $\sigma^{2}$-sub-Gaussian noises $\lbp \eta_{n}: n \in \lb T \rb \rbp$ are independent but not identically distributed.
\end{remark}

As Corollary \ref{cor:m.d} suggests, the GLR test in \eqref{eq:GLR_tau-pre-post-unknown} and the GSR test in \eqref{eq:GSR_tau-pre-post-unknown} are order-optimal with respect to $T$ when $m = \mcal{O}\lp \log T \rp$. Since these tests require minimal prior knowledge about the pre- and post-change distributions in operation, we can apply them to practical scenarios such as change detection in piecewise stationary bandits \citep{huang2025change, gerogiannis2025detectionneedfeasibleoptimal}.

\section{Experimental Results}
\label{sec:sim}

In this section, we illustrate the performance of our change detectors through simulations. We present the empirical latencies of the TVT-CuSum, TVT-SR, and GLR tests and compare them with the lower bound in Theorem \ref{thm:low_bound_latency}. We also compare the empirical values with the upper bounds on the latencies in Theorems \ref{thm:latency_pre-post-known} and \ref{thm:latency_pre-post-unknown}, to assess the tightness of the bounds. 

To study the empirical performance, we first explain how the empirical values are determined. From the definition of latency in \eqref{eq:latency}, when the change-point lies in $\lb m+1, T-\ell \rb$, the probability of detection delay surpassing the latency should be at most $\delta_{\mrm{D}}$. Accordingly, we compute the empirical latencies as follows: For each change-point in a candidate change-point set $\mcal{M} \subseteq \lb m+1, T-\ell \rb$, we conduct $2\times10^{5}$ trials and record the detection delay $\tau - \nu$ in each trial. The empirical value is then set to the maximum, over all change-points in $\mcal{M}$, of the $100 \lp 1-\delta_{\mrm{D}} \rp^{\mrm{th}}$ percentiles of the $2\times10^{5}$ recorded delays. Furthermore, we take $\mcal{M}$ to be $\lbp m+1+nT/10 :\; n \in \mbb{N},\; m+1+nT/10 \leq T \rbp$, because conducting $2\times10^5$ trials over all $T-\ell-m$ change-points is computationally cumbersome. 

We now describe the experimental setup: The pre-change distribution is $\mcal{N}\lp 0, 1 \rp$, and the post-change distribution is $\mcal{N}\lp 1, 1 \rp$. For the GLR test in \eqref{eq:GLR_tau-pre-post-unknown}, we set the pre-change window length $m = T - 1000$, which is larger than $\lce \frac{16\sigma^{2}}{\Delta^{2}}\beta_{\mrm{GLR}} \lp T,\delta_{\mrm{F}}\rp \rce$ in Corollary \ref{cor:m.d}, ensuring that the latency upper bound in Theorem \ref{thm:latency_pre-post-unknown} is $\mcal{O}\lp \log T \rp$. 
By exploiting the recursions in \eqref{eq:CuSum_recur} and \eqref{eq:SR_recur}, we can compute the CuSum and SR statistics efficiently and implement the TVT-CuSum and TVT-SR tests. 
However, the GLR statistic in  \eqref{eq:GLR_stat_pre-post-unknown} and the GSR statistic in \eqref{eq:GSR_stat_pre-post-unknown} do not have a recursive structure, making direct computation intractable. Therefore, we perform downsampling by taking the supremum over $\mcal{K}_{n} \coloneqq \lb n - 700, n \rb$ when computing the GLR statistic in \eqref{eq:GLR_stat_pre-post-unknown}, i.e., for $n \in \lb T \rb$,
\begin{align} 
    G'_{n} &\coloneqq \sup_{k \in \mcal{K}_{n}} \log \frac{\sup_{\mu_{0}' \in \mbb{R}} \sup_{\mu_{1}' \in \mbb{R}} \prod_{i = 1}^{k} f_{\mu_{0}'}\lp X_{i} \rp \prod_{i = k + 1}^{n} f_{\mu_{1}'}\lp X_{i} \rp}{\sup_{\mu\in\mbb{R}} \prod_{i = 1}^{n} f_{\mu} \lp X_{i} \rp} .\label{eq:GLR_stat_pre-post-unknown-down}
\end{align}
The stopping time of the implemented GLR test is then:
\begin{align}
\tau'_{\mrm{GLR}}&\coloneqq\inf\lbp n\in\mbb{N}:\; G'_{n} \geq\beta_{\mrm{GLR}}\lp n,\delta_{\mathrm{F}}\rp\rbp \label{eq:GLR_tau_pre-post-unknown-down},
\end{align}
where $\beta_{\mrm{GLR}}$ is defined in \eqref{eq:GLR_thres-pre-post-unknown}. Unfortunately, such downsampling cannot be applied to the computation of the GSR statistic in \eqref{eq:GSR_stat_pre-post-unknown}. Consequently, we report only the empirical performance of the GLR test in  \eqref{eq:GLR_tau_pre-post-unknown-down}. The empirical results appear in Figure~\ref{fig:latency_T}, where the empirical values of our tests increase linearly with $\log T$. 
%
\begin{figure}
    \centering
    \includegraphics[width=0.9\linewidth]{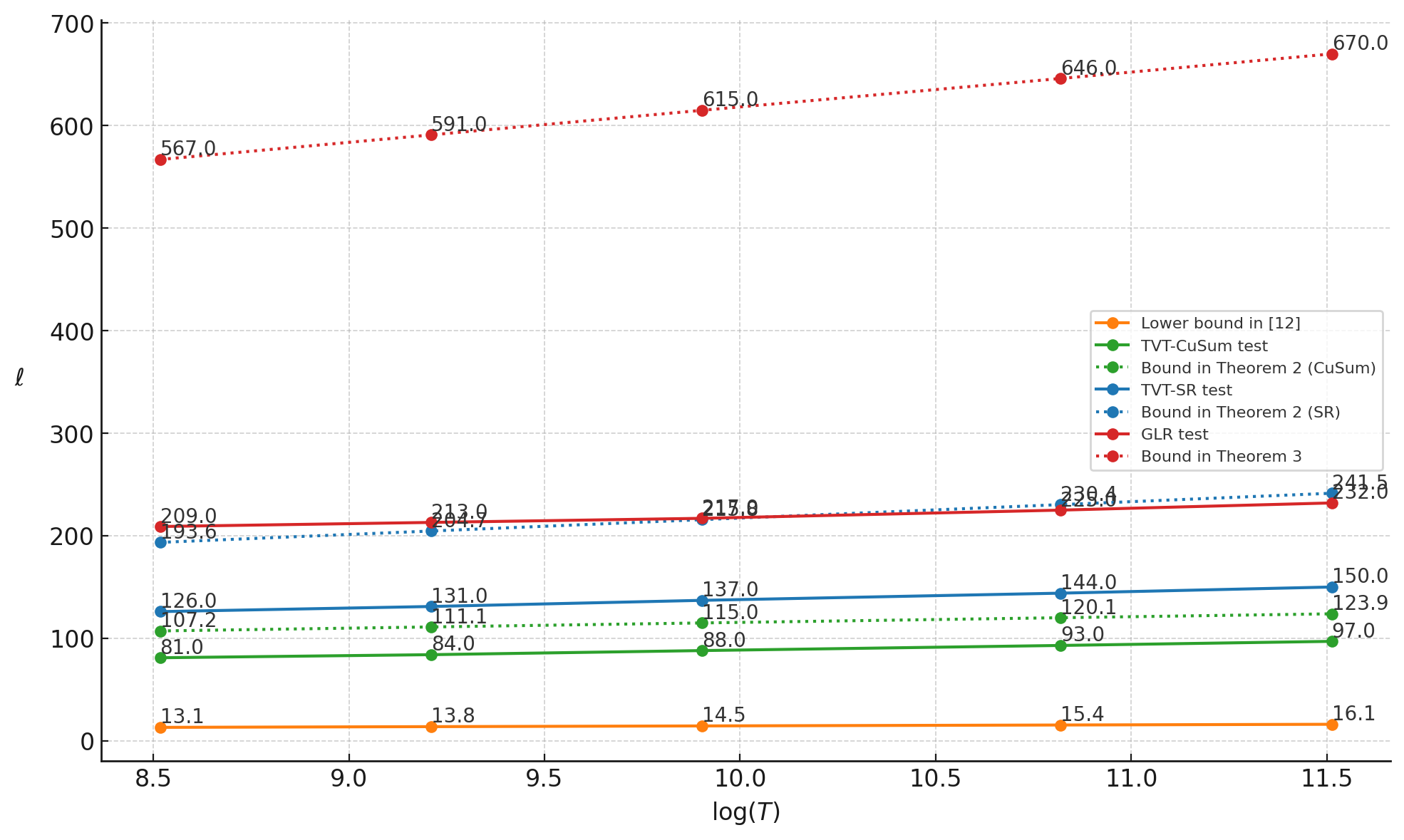}
    \caption{Latencies of TVT-CuSum, TVT-SR, and GLR tests with fixed $\delta_{\mrm{F}} = \delta_{\mrm{D}} = 0.01$ and varying $T \in \lbp 5000, 10000, 20000, 50000, 100000 \rbp$.}
    \label{fig:latency_T}
\end{figure}

As shown in Figure \ref{fig:latency_T}, the empirical latencies are linear with respect to $\log T$, corroborating the order-optimal performance of the tests. The gap between the upper bound in Theorem \ref{thm:latency_pre-post-unknown} and the empirical latencies of the GLR test is significantly larger than the gaps between the upper bounds in Theorem \ref{thm:latency_pre-post-known} and the empirical latencies of the TVT-CuSum and TVT-SR tests. This indicates that it may be possible to derive a tighter upper bound on the latency of the GLR test. In addition, the figure also indicates that the lower bound is also loose. There are two potential reasons for this looseness: (i) in deriving the lower bound, we did not impose the condition that test $\tau$ is oblivious to the knowledge of the horizon $T$; and (ii) the TVT-CuSum test could be improved considerably. We believe that the first reason is more likely to be true.

\section{Conclusions}
\label{sec:Sum}

In this work, we proposed a variant of the QCD problem, in which the goal is to minimize the latency (at a prespecified latency level), while controlling the false alarm probability. We established a lower bound on the latency, under prespecified false alarm and latency levels, demonstrating that the latency of an optimal change detector is at least $\Omega(\log T)$. We then developed two change detectors that leverage the full knowledge of the pre- and post-change distributions, and then generalized these tests to cases with minimal knowledge of the distributions. All the proposed tests attain $\mcal{O}\lp\log T\rp$ latency while successfully containing the false alarm probability. 
Our simulations validate the order optimality of the tests with respect to the horizon, and also show that our tests achieve good latency performance.
As the generalized tests can be applied to scenarios where little prior information is available about the pre- and post-change distributions, these tests can be employed in a wide range of practical settings. 

However, the simulation results also suggest that the lower bound on the latency in Theorem \ref{thm:low_bound_latency} and the upper bound on the latency in Theorem \ref{thm:latency_pre-post-unknown} are quite loose, and it may be worthwhile exploring if these bounds can be tightened. In addition, the gap between the latencies of the tests that utilize the knowledge of the distribution and that of the generalized versions of these tests indicates that it may be worthwhile attempting to design tests that improve upon the generalized tests. We leave these topics for future exploration.

\section*{Acknowledgement}
This work was supported by the National Science Foundation under grant ECCS-2033900, and by the Army Research Laboratory under Cooperative Agreement W911NF-17-2-0196, through the University of Illinois at Urbana-Champaign.

\bibliography{ref.bib}

\onecolumn

\appendix

In the appendix, we first present the proof of Theorem \ref{thm:latency_pre-post-known}, and then the proof of Theorem \ref{thm:low_bound_latency}. The reason for this order of presentation is because some steps in the proof of Theorem \ref{thm:low_bound_latency} require results that are in the proof of Theorem \ref{thm:latency_pre-post-known}. 

\section{Proof of Theorem \ref{thm:latency_pre-post-known}} \label{sec:thm2}

There are two change detectors to consider in Theorem \ref{thm:latency_pre-post-known}: the TVT-CuSum and TVT-SR tests. For each test, we need to establish that the false alarm probability, $\Pr_{\infty} \lp \tau \leq T \rp$, is upper-bounded by $\delta_{\mrm{F}}$, and late detection probability, $\Pr_{\nu} \lp \tau \geq \nu + \ell \rp$, is upper-bounded by $\delta_{\mrm{D}}$. For the TVT-CuSum test, Ville's inequality \citep{ville1939etude} leads to the false alarm probability upper bound, while the Chernoff bound  leads to the late detection probability upper bound. We then use the false alarm probability upper bound for the TVT-CuSum test to prove that of the TVT-SR test. 

First, consider the TVT-CuSum test with a fixed $T \in \mbb{N}$ and $\delta_{\mrm{F}}, \delta_{\mrm{D}} \in \lp 0, 1 \rp$. We can upper bound the probability of false alarm as follows: for all $T\in\mbb{N}$ and $r > 1$,
%
\begin{align}
&\Pr_{\infty}\lp\tau_{\mrm{C},r}\leq T\rp \nonumber\\
&\leq\Pr_{\infty}\lp\tau_{\mrm{C},r}<\infty\rp \nonumber\\
&=\Pr_{\infty}\lp\exists\,n\geq1:\;C_{n}\geq\beta_{\mrm{C}}\lp n,\delta_{\mathrm{F}},r\rp\rp \nonumber\\
&=\Pr_{\infty}\lp\exists\,n\geq1,\,\exists\,j\in\lb n\rb:\; \sum_{i=j}^{n}\log\lp\frac{f_{1}\lp X_{i}\rp}{f_{0}\lp X_{i}\rp}\rp\geq\log\lp\zeta\lp r\rp\frac{n^{r}}{\delta_{\mathrm{F}}}\rp\rp \nonumber\\
&=\Pr_{\infty}\lp\exists\,n\geq1,\,\exists\,j\in\lb n\rb:\; \prod_{i=j}^{n}\frac{f_{1}\lp X_{i}\rp}{f_{0}\lp X_{i}\rp}\geq\zeta\lp r\rp\frac{n^{r}}{\delta_{\mathrm{F}}}\rp\nonumber\\
&=\Pr_{\infty}\left(\exists\,j\geq1,\;\exists\,k\geq0:\;\prod_{i=j}^{j+k}\frac{f_{1}\lp X_{i}\rp}{f_{0}\lp X_{i}\rp}\geq\zeta\lp r\rp\frac{\lp j+k\rp^{r}}{\delta_{\mathrm{F}}}\right)\nonumber\\
&\overset{\aaaa}{\leq}\sum_{j=1}^{\infty} \Pr_{\infty}\left(\exists\,k\geq0:\;\frac{1}{\lp j+k\rp^{r}}\prod_{i=j}^{j+k}\frac{f_{1}\lp X_{i}\rp}{f_{0}\lp X_{i}\rp}\geq\frac{\zeta\lp r\rp}{\delta_{\mathrm{F}}}\right) \label{eq:TVT-CuSum_FA_1}
\end{align}
where step $(a)$ results from union bound. In the following lemma, we demonstrate that the sequence $\lbp\frac{1}{\lp j+k\rp^{r}}\prod_{i=j}^{j+k}\frac{f_{1}\lp X_{i}\rp}{f_{0}\lp X_{i}\rp}\rbp_{k=0}^{\infty}$ is a supermartingale, allowing us to use Ville's inequality \citep{ville1939etude}.

\begin{lemma}\label{lem:sup_martin} For any 
$r>1$ and $j \in \mbb{N}$, under the probability measure $\Pr_{\infty}$, i.e., $X_{k}\diid f_{0}$, the sequence 
\begin{align}
\lbp\frac{1}{\lp j+k\rp^{r}}\prod_{i=j}^{j+k}\frac{f_{1}\lp X_{i}\rp}{f_{0}\lp X_{i}\rp}\rbp_{k=0}^{\infty}  
\end{align}
%
is a nonnegative supermartingale with respect to $\lbp \mathcal{F}_{k} \rbp_{k=-1}^{\infty}$, where $\mathcal{F}_{-1}$ is the trivial sigma-algebra and $\mathcal{F}_{k} \coloneqq \sigma \lp X_{j}, \dots, X_{j+k} \rp$ for each $k \geq 0$.
\end{lemma}

\begin{proof}
Recall that $\lambda$ is the dominating measure of $f_{0}$ and $f_{1}$. We can observe that for any $k\geq 0$:
\begin{align}\nonumber
&\E_{\infty}\lb\frac{1}{\lp j+k\rp^{r}}\prod_{i=j}^{j+k}\frac{f_{1}\lp X_{i}\rp}{f_{0}\lp X_{i}\rp}\bigg|\mcal{F}_{k-1}\rb\\\nonumber
&=\frac{1}{\lp j+k\rp^{r}}\prod_{i=j}^{j+k-1}\frac{f_{1}\lp X_{i}\rp}{f_{0}\lp X_{i}\rp}\E_{\infty}\lb\frac{f_{1}\lp X_{j+k}\rp}{f_{0}\lp X_{j+k}\rp}\bigg|\mcal{F}_{k-1}\rb\\\nonumber
&\overset{(a)}{=}\frac{1}{\lp j+k\rp^{r}}\prod_{i=j}^{j+k-1}\frac{f_{1}\lp X_{i}\rp}{f_{0}\lp X_{i}\rp}\E_{\infty}\lb\frac{f_{1}\lp X_{j+k}\rp}{f_{0}\lp X_{j+k}\rp}\rb\\\nonumber
&=\frac{1}{\lp j+k\rp^{r}}\prod_{i=j}^{j+k-1}\frac{f_{1}\lp X_{i}\rp}{f_{0}\lp X_{i}\rp}\int_{\mbb{R}}\frac{f_{1}\lp x\rp}{f_{0}\lp x\rp}f_{0}\lp x\rp d\lambda\lp x\rp\\\nonumber
&=\frac{1}{\lp j+k\rp^{r}}\prod_{i=j}^{j+k-1}\frac{f_{1}\lp X_{i}\rp}{f_{0}\lp X_{i}\rp}\\
&\leq\frac{1}{\lp j+k-1\rp^{r}}\prod_{i=j}^{j+k-1}\frac{f_{1}\lp X_{i}\rp}{f_{0}\lp X_{i}\rp}\label{eq:sup_mart}
\end{align}
%
where step $(a)$ is due to the independence of $\lbp X_{n}\rbp_{n=1}^{\infty}$. Since each $\frac{1}{\lp j+k\rp^{r}}\prod_{i=j}^{j+k}\frac{f_{1}\lp X_{i}\rp}{f_{0}\lp X_{i}\rp}$ is nonnegative, Lemma \ref{lem:sup_martin} follows from \eqref{eq:sup_mart}.
\end{proof}

With Lemma \ref{lem:sup_martin}, we can apply Ville's inequality to \eqref{eq:TVT-CuSum_FA_1} and obtain:
\begin{equation}\label{eq:TVT-CuSum_FA_2}
\begin{aligned}
\Pr_{\infty}\lp\tau_{\mrm{C},r}\leq T\rp
\leq\sum_{j=1}^{\infty}\frac{\delta_{\mathrm{F}}}{\zeta\lp r\rp}\E_{\infty}\lb\frac{1}{j^{r}}\frac{f_{1}\lp X_{j}\rp}{f_{0}\lp X_{j}\rp}\rb =\frac{\delta_{\mathrm{F}}}{\zeta\lp r\rp}\sum_{j=1}^{\infty}\frac{1}{j^{r}} =\delta_{\mathrm{F}}.
\end{aligned}
\end{equation}
This completes the proof of the false alarm probability upper bound for the TVT-CuSum test.

Next, we upper bound the late detection probability of the TVT-CuSum test: It is easy to show that $\Lambda\lp 0 \rp=\Lambda\lp1\rp=0$.
%
%
Since $f_{1}\neq f_{0}$,  $\Lambda$ is strictly convex, and therefore, $\Lambda\lp\theta\rp<0$ for any $\theta\in\lp0,1\rp$. Then, by the definition of $\tau_{\mrm{C}, r}$ in \eqref{eq:TVT-CuSum_tau}, for any $\nu \in \lb T-\ell\rb$,
\begin{align}
    \Pr_{\nu}\lp\tau_{\mrm{C}, r}  \geq \nu+\ell\rp 
    &=\Pr_{\nu}\lp\inf\lbp n\in\mbb{N}:\; C_{n}\geq\beta_{\mrm{C}}\lp n,\delta_{\mathrm{F}},r\rp\rbp\geq\nu+\ell\rp\nonumber\\
&=\Pr_{\nu}\left(\forall\,n\in\lb\nu+\ell-1\rb:\;C_{n}<\beta_{\mrm{C}}\lp n,\delta_{\mathrm{F}},r\rp\right).\label{eq:TVT-CuSum_LD_1}
\end{align}
Let $\theta = \argmin_{\theta' \in \lp 0,1 \rp} \lbp\lb\log\lp1/\delta_{\mathrm{D}}\rp+\theta'\beta_{\mrm{C}}\lp T, \delta_{\mrm{F}}, r\rp\rb/\lvert\Lambda\lp\theta'\rp\rvert\rbp$. Then, we have
\begin{align}
&\Pr_{\nu}\lp\tau_{\mrm{C}, r} \geq \nu+\ell\rp \nonumber \\ 
&\overset{\aaaa}{\leq}\Pr_{\nu}\lp C_{\nu+\ell-1}<\beta_{\mrm{C}}\lp \nu+\ell-1,\delta_{\mathrm{F}},r\rp\rp \nonumber\\
&=\Pr_{\nu}\Bigg(\max_{j \in \lb\nu+\ell-1\rb}\sum_{i=j}^{\nu+\ell-1}\log\lp\frac{f_{1}\lp X_{i}\rp}{f_{0}\lp X_{i}\rp}\rp<\beta_{\mrm{C}}\lp \nu+\ell-1,\delta_{\mathrm{F}},r\rp\Bigg)\nonumber\\
&\overset{\bbbb}{\leq}\Pr_{\nu}\lp\sum_{i=\nu}^{\nu+\ell-1}\log\lp\frac{f_{1}\lp X_{i}\rp}{f_{0}\lp X_{i}\rp}\rp<\beta_{\mrm{C}}\lp \nu+\ell-1,\delta_{\mathrm{F}},r\rp\rp\nonumber\\
&=\Pr_{\nu}\lp-\sum_{i=\nu}^{\nu+\ell-1}\log\lp\frac{f_{1}\lp X_{i}\rp}{f_{0}\lp X_{i}\rp}\rp>-\beta_{\mrm{C}}\lp \nu+\ell-1,\delta_{\mathrm{F}},r\rp\rp\nonumber\\
&\overset{\cccc}{\leq}\exp\lp-\ell\lp-\frac{1}{\ell}\theta\log\lp\beta_{\mrm{C}}\lp \nu+\ell-1,\delta_{\mathrm{F}},r\rp\rp-\Lambda\lp\theta\rp\rp\rp\nonumber\\
&=\exp\lp\theta\beta_{\mrm{C}}\lp \nu+\ell-1,\delta_{\mathrm{F}},r\rp+ \ell\Lambda\lp\theta\rp\rp\nonumber \\
&\overset{\dddd}{\leq}\exp\lp\theta\beta_{\mrm{C}}\lp T,\delta_{\mathrm{F}},r\rp+ \ell\Lambda\lp\theta\rp\rp\nonumber \\
&\overset{\eeee}{=}\delta_{\mrm{D}}\label{eq:TVT-CuSum_LD_2}
\end{align}

where step $(a)$ is due to the fact that $\lbp \nu+\ell-1\rbp\subseteq\lb \nu+\ell-1\rb$, while step $(b)$ is owing to the fact that $\sum_{i=\nu}^{\nu+\ell-1}\log\lp\frac{f_{1}\lp X_{i}\rp}{f_{0}\lp X_{i}\rp}\rp\leq\max_{j \in \lb\nu+\ell-1\rb}\sum_{i=j}^{\nu+\ell-1}\log\lp\frac{f_{1}\lp X_{i}\rp}{f_{0}\lp X_{i}\rp}\rp$. Step $(c)$ stems from the Chernoff bound \citep{chernoff1952measure}, and step $(d)$ results from the fact that $\nu+\ell-1\leq T$. In step $(e)$, we plug in the definition of $\ell$ in \eqref{eq:latency_pre-post-known}. This completes the proof of the late detection probability upper bound for the TVT-CuSum test.

Now, consider the TVT-SR test with a fixed $T \in \mbb{N}$ and $\delta_{\mrm{F}}, \delta_{\mrm{D}} \in \lp 0, 1 \rp$. We can upper bound the probability of false alarm as follows: for all $T\in\mbb{N}$ and $r > 1$,
\begin{align}
&\Pr_{\infty}\lp\tau_{\mrm{S},r}\leq T\rp \nonumber\\
&=\Pr_{\infty}\lp\exists\,n\in\lb T \rb:\;\log S_{n}\geq\beta_{\mrm{S}}\lp n,\delta_{\mathrm{F}},r\rp\rp \nonumber\\
&=\Pr_{\infty}\lp \exists\,n\in\lb T \rb:\; \log\lp\sum_{j=1}^{n} \prod_{i=j}^{n} \frac{f_{1}\lp X_{i}\rp}{f_{0}\lp X_{i}\rp}\rp\geq\beta_{\mrm{C}}\lp n,\delta_{\mathrm{F}},r\rp + \log n\rp \nonumber\\
&=\Pr_{\infty}\lp \exists\,n\in\lb T \rb:\; \sum_{j=1}^{n}\prod_{i=j}^{n}\frac{f_{1}\lp X_{i}\rp}{f_{0}\lp X_{i}\rp}\geq n\exp\lp\beta_{\mrm{C}}\lp n,\delta_{\mathrm{F}},r\rp\rp\rp\nonumber\\
&=\Pr_{\infty}\lp \exists\,n\in\lb T \rb:\; \frac{1}{n}\sum_{j=1}^{n}\prod_{i=j}^{n}\frac{f_{1}\lp X_{i}\rp}{f_{0}\lp X_{i}\rp}\geq \exp\lp\beta_{\mrm{C}}\lp n,\delta_{\mathrm{F}},r\rp\rp\rp\nonumber\\
&\overset{\aaaa}{\leq}\Pr_{\infty}\lp \exists\,n\in\lb T \rb:\; \max_{j \in \lb n \rb} \prod_{i=j}^{n}\frac{f_{1}\lp X_{i}\rp}{f_{0}\lp X_{i}\rp}\geq \exp\lp\beta_{\mrm{C}}\lp n,\delta_{\mathrm{F}},r\rp\rp \rp\nonumber\\
&=\Pr_{\infty}\lp \exists\,n\in\lb T \rb:\; C_{n} \geq \beta_{\mrm{C}}\lp n,\delta_{\mathrm{F}},r\rp \rp\nonumber\\
&=\Pr_{\infty}\lp\tau_{\mrm{C},r}\leq T\rp\nonumber\\
&\overset{\bbbb}{\leq}\delta_{\mrm{F}} \label{eq:TVT-SR_FA_1}
\end{align}
where step $(a)$ results from the fact that $\frac{1}{n}\sum_{j=1}^{n}\prod_{i=j}^{n}\frac{f_{1}\lp X_{i}\rp}{f_{0}\lp X_{i}\rp} \leq \max_{j \in \lb n \rb} \prod_{i=j}^{n}\frac{f_{1}\lp X_{i}\rp}{f_{0}\lp X_{i}\rp}$, and step $(b)$ stems from \eqref{eq:TVT-CuSum_FA_2}. This completes the proof of the false alarm probability upper bound for the TVT-SR test.

Next, we upper bound the late detection probability of the TVT-SR test: Let $\theta = \argmin_{\theta' \in \lp 0,1 \rp} \lbp\lb\log\lp1/\delta_{\mathrm{D}}\rp+\theta'\beta_{\mrm{S}}\lp T, \delta_{\mrm{F}}, r\rp\rb/\lvert\Lambda\lp\theta'\rp\rvert\rbp$. By the definition of $\tau_{\mrm{S}, r}$ in \eqref{eq:TVT-SR_tau}, for any $\nu \in \lb T-\ell\rb$,
\begin{align}
    &\Pr_{\nu}\lp\tau_{\mrm{S}, r} \geq \nu+\ell\rp \nonumber\\
    &=\Pr_{\nu}\lp\inf\lbp n\in\mbb{N}:\; \log S_{n}\geq\beta_{\mrm{S}}\lp n,\delta_{\mathrm{F}},r\rp\rbp\geq\nu+\ell\rp \nonumber\\
    &=\Pr_{\nu}\lp\forall\,n\in\lb\nu+\ell-1\rb:\;\sum_{j=1}^{n} \prod_{i=j}^{n}\frac{f_{1}\lp X_{i}\rp}{f_{0}\lp X_{i}\rp}<\exp\lp\beta_{\mrm{S}}\lp n,\delta_{\mathrm{F}},r\rp\rp\rp\nonumber\\
    &\overset{\aaaa}{\leq} \Pr_{\nu}\lp\sum_{j=1}^{\nu+\ell-1} \prod_{i=j}^{\nu+\ell-1}\frac{f_{1}\lp X_{i}\rp}{f_{0}\lp X_{i}\rp}<\exp\lp\beta_{\mrm{S}}\lp \nu+\ell-1,\delta_{\mathrm{F}},r\rp\rp\rp \nonumber\\
    &\overset{\bbbb}{\leq} \Pr_{\nu}\lp \prod_{i=\nu}^{\nu+\ell-1}\frac{f_{1}\lp X_{i}\rp}{f_{0}\lp X_{i}\rp}<\exp\lp\beta_{\mrm{S}}\lp \nu+\ell-1,\delta_{\mathrm{F}},r\rp\rp\rp \nonumber\\
    &=\Pr_{\nu}\lp-\sum_{i=\nu}^{\nu+\ell-1}\log\lp\frac{f_{1}\lp X_{i}\rp}{f_{0}\lp X_{i}\rp}\rp>-\beta_{\mrm{S}}\lp \nu+\ell-1,\delta_{\mathrm{F}},r\rp\rp \nonumber\\
    &\overset{\cccc}{\leq}\exp\lp-\ell\lp-\frac{1}{\ell}\theta\beta_{\mrm{S}}\lp \nu+\ell-1,\delta_{\mathrm{F}},r\rp-\Lambda\lp\theta\rp\rp\rp \nonumber\\
    &=\exp\lp \theta\beta_{\mrm{S}}\lp \nu+\ell-1,\delta_{\mathrm{F}},r\rp + \ell\Lambda\lp\theta\rp\rp \nonumber\\
    &\overset{\dddd}{\leq}\exp\lp \theta\beta_{\mrm{S}}\lp T,\delta_{\mathrm{F}},r\rp + \ell\Lambda\lp\theta\rp\rp \nonumber\\
    &\overset{\eeee}{=}\delta_{\mrm{D}} \label{eq:TVT-SR_LD_1}
\end{align}
where step $(a)$ is due to the fact that $\lbp \nu+\ell-1\rbp\subseteq\lb\nu+\ell-1\rb$, while step $(b)$ is owing to the fact that $\sum_{j=1}^{\nu+\ell-1} \prod_{i=j}^{\nu+\ell-1}\frac{f_{1}\lp X_{i}\rp}{f_{0}\lp X_{i}\rp} \leq \prod_{i=\nu}^{\nu+\ell-1}\frac{f_{1}\lp X_{i}\rp}{f_{0}\lp X_{i}\rp}$. Step $(c)$ stems from the Chernoff bound \citep{chernoff1952measure}, and step $(d)$ results from the fact that $\nu+\ell-1\leq T$. In step $(e)$, we plug in the definition of $\ell$ in \eqref{eq:latency_pre-post-known}. This completes the proof of the late detection probability upper bound for the TVT-SR test.

\section{Proof of Theorem \ref{thm:low_bound_latency}} \label{sec:thm1}

Consider a fixed arbitrary $T \in \mbb{N}$ and  $\delta_{\mrm{F}}, \delta_{\mrm{D}} \in \lp 0, 1 \rp$. Let $\tau$ and $\ell$ be an arbitrary stopping time and its latency satisfying $\Pr_{\infty}\lp \tau \leq T \rp \leq \delta_{\mrm{F}}$ and $\Pr_{\nu}\lp \tau \geq \nu+\ell \rp \leq \delta_{\mrm{D}}$ for any $\nu \in \lb 1, T-\ell\rb$. To prove Theorem \ref{thm:low_bound_latency}, we partition the event $\lbp \tau\geq\nu+\ell \rbp$ into three disjoint events, two of which are bounded by change of measure and one of which is bounded by Markov inequality. For any $c > K$, define the events
\begin{equation}
\mcal{A}\coloneqq\lbp\nu\leq\tau<\nu+\ell,\sum_{i=\nu}^{\nu+\ell-1}\!\!\log\lp\frac{f_{1}\lp X_{i}\rp}{f_{0}\lp X_{i}\rp}\rp\geq \ell c\rbp
\end{equation}
and 
\begin{equation}\label{eq:Bc}
\mcal{B}\coloneqq\lbp\nu\leq\tau<\nu+\ell,\sum_{i=\nu}^{\nu+\ell-1}\!\!\log\lp\frac{f_{1}\lp X_{i}\rp}{f_{0}\lp X_{i}\rp}\rp<\ell c\rbp.
\end{equation}
We note that $\mcal{A}\cap \mcal{B} = \emptyset$ and $\mcal{A}\cup \mcal{B}= \lbp \nu\leq \tau < \nu+\ell \rbp$, which we will use later in the proof. From the problem formulation \eqref{eq:QCD} we have that for any $\nu\in\lb T-\ell \rb$, 
\begin{align}
\delta_{\mathrm{D}}
&\geq\Pr_{\nu}\lp\tau\geq\nu+\ell\rp\nonumber\\
&=1-\Pr_{\nu}\lp\tau<\nu+\ell\rp \nonumber\\
&=1-\Pr_{\nu}\lp\tau<\nu\rp-\Pr_{\nu}\lp\nu\leq\tau<\nu+\ell\rp \nonumber\\
&=1-\Pr_{\infty}\lp\tau<\nu\rp-\Pr_{\nu}\lp\nu\leq\tau<\nu+\ell\rp \nonumber\\
&\geq1-\Pr_{\infty}\lp\tau\leq T\rp-\Pr_{\nu}\lp\mcal{B}\rp-\Pr_{\nu}\lp\mcal{A}\rp.\label{eq:thm_1_sketch_1}
\end{align}

Next, we derive an upper bound on $\Pr_{\nu}\lp\mcal{B}\rp$ by using change of measure and the following lemma:

\begin{lemma}\label{lem:LowBoundFA} For any stopping time $\tau$ that satisfies the false alarm probability constraint $\Pr_{\infty} \lp \tau \leq T \rp \leq \delta_{\mrm{F}}$, any $\delta_{\mathrm{F}}\in\lp0,1\rp$, and any $d\in\lb T\rb$, there exists a time-step $\tilnu\in\lb T-d\rb\cup\lbp 0\rbp$ such that:
\begin{align}
\Pr_{\infty}\lp\tilnu\leq\tau< \tilnu+d\rp\leq\frac{\delta_{\mathrm{F}}}{\lfl T/d\rfl}.\label{eq:FA_window}
\end{align}
\end{lemma}

\begin{proof}
Suppose that there exists a $d \in \lb T\rb$ such that for any $\nu\in\lb T-d\rb$,
\begin{align}
\Pr_{\infty}\lp\nu\leq\tau<\nu+d\rp>\frac{\delta_{\mathrm{F}}}{\lfl T/d\rfl}.
\end{align}
Then,
\begin{align}
\Pr_{\infty}\lp\tau\leq T\rp
&\geq\Pr_{\infty}\lp\tau< d\lfl T/d\rfl\rp
=\!\!\sum_{i=0}^{\lfl T/d\rfl-1}\!\!\Pr_{\infty}\lp d i\leq\tau
<d\lp i+1\rp\rp>\!\!\sum_{i=0}^{\lfl T/d \rfl-1}\!\!\frac{\delta_{\mathrm{F}}}{\lfl T/d \rfl}
=\delta_{\mathrm{F}}.
\end{align}
This leads to a contradiction since $\Pr_{\infty}\lp\tau\leq T\rp\leq\delta_{\mathrm{F}}$.
\end{proof}

Since $\tau$ satisfies the false alarm probability constraint $\Pr_{\infty}\lp\tau\leq T\rp\leq\delta_{\mathrm{F}}$, by Lemma \ref{lem:LowBoundFA}, there exists a change-point $\tilnu \in\lb T-\ell\rb$ such that
\begin{align}
    \Pr_{\infty}\lp\tilnu\leq\tau < \tilnu+\ell\rp\leq\frac{\delta_{\mathrm{F}}}{\lfl T/\ell\rfl}.\label{eq:FA_window_tilnu}
\end{align}
Recall that $\lambda$ is the dominating measure of $f_{0}$ and $f_{1}$. For this choice of $\tilnu$, we have
\begin{align}\nonumber
\Pr_{\tilnu}\lp\mcal{B}\rp&=\int_{\mcal{B}}\prod_{i=1}^{\tilnu-1}f_{0}\lp x_{i}\rp\prod_{i=\tilnu}^{\tilnu+\ell-1}f_{1}\lp x_{i}\rp\otimes_{i=1}^{\tilnu+\ell-1} d \lambda\lp x_{i}\rp\\\nonumber  &=\int_{\mcal{B}}\prod_{i=\tilnu}^{\tilnu+\ell-1}\frac{f_{1}\lp x_{i}\rp}{f_{0}\lp x_{i}\rp}\prod_{i=1}^{\tilnu+\ell-1}f_{0}\lp x_{i}\rp\otimes_{i=1}^{\tilnu+\ell-1} d \lambda\lp x_{i}\rp\\\nonumber
&\overset{\aaaa}{\leq}\int_{\mcal{B}}\exp\lp \ell c\rp\prod_{i=1}^{\tilnu+\ell-1}f_{0}\lp x_{i}\rp\otimes_{i=1}^{\tilnu+\ell-1} d \lambda\lp x_{i}\rp\\\nonumber
&=\exp\lp \ell c\rp\int_{\mcal{B}}\prod_{i=1}^{\tilnu+\ell-1}f_{0}\lp x_{i}\rp\otimes_{i=1}^{\tilnu+\ell-1} d \lambda\lp x_{i}\rp\\\nonumber
&\overset{\bbbb}{=}\exp\lp \ell c\rp\Pr_{\infty}\lp\mcal{B}\rp\\\nonumber
&\overset{\cccc}{\leq}\exp\lp \ell c\rp\Pr_{\infty}\lp\tilnu\leq\tau<\tilnu+\ell\rp\\
&\overset{\dddd}{\leq}\frac{\exp\lp \ell c\rp\delta_{\mathrm{F}}}{\lfl T/\ell\rfl}\label{eq:thm_1_sketch_2}
\end{align}
where step $(a)$ results from the definition of $\mcal{B}$ in \eqref{eq:Bc}, and step $(b)$ stems from the fact that under $\Pr_{\infty}$, every $X_{i}$ follows the density $f_{0}$. Step $(c)$ is owing to the fact that $\mcal{B}\subseteq\lbp\tilnu\leq\tau<\tilnu+\ell\rbp$, and step $(d)$ is due to \eqref{eq:FA_window_tilnu}.

Then, we derive an upper bound on $\Pr_{\nu}\lp\mcal{A}\rp$: for any $\nu\in\lb T-\ell \rb$,
\begin{align}
\Pr_{\nu}\lp\mcal{A}\rp
&=\Pr_{\nu}\lp\nu\leq\tau<\nu+\ell,\sum_{i=\nu}^{\nu+\ell-1}\log\lp\frac{f_{1}\lp X_{i}\rp}{f_{0}\lp X_{i}\rp}\rp\geq \ell c\rp \nonumber\\
&\leq\Pr_{\nu}\lp\sum_{i=\nu}^{\nu+\ell-1}\log\lp\frac{f_{1}\lp X_{i}\rp}{f_{0}\lp X_{i}\rp}\rp\geq \ell c\rp \nonumber\\
&=\Pr_{\nu}\lp\prod_{i=\nu}^{\nu+\ell-1}\frac{f_{1}\lp X_{i}\rp}{f_{0}\lp X_{i}\rp}\geq e^{\ell c}\rp \nonumber\\
&\overset{(a)}{\leq}e^{-\ell c}\E_{\nu}\lb\prod_{i=\nu}^{\nu+\ell-1}\frac{f_{1}\lp X_{i}\rp}{f_{0}\lp X_{i}\rp}\rb \nonumber\\
&=e^{-\ell c}\prod_{i=\nu}^{\nu+\ell-1}\E_{\nu}\lb\frac{f_{1}\lp X_{i}\rp}{f_{0}\lp X_{i}\rp}\rb \nonumber\\
&\overset{(b)}{\leq}e^{-\ell\lp c-K\rp}\label{eq:thm_1_sketch_3}
\end{align}
where step $(a)$ results from the Markov inequality.
Step $(b)$ is due to the definition of $K$ in \eqref{eq:C_const}.

Next, define $s \coloneqq \lp T - \ell \rp/T$, $\tilde{c} \coloneqq c + \log\lp \ell\rp / \ell$, and $\varepsilon \coloneqq e^{-\ell\lp c-K\rp}$. By plugging \eqref{eq:thm_1_sketch_2} and \eqref{eq:thm_1_sketch_3} into \eqref{eq:thm_1_sketch_1} with $\nu=\tilnu$, we have 
\begin{align}
\delta_{\mathrm{D}}
&\geq1-\delta_{\mathrm{F}}-\frac{\exp\lp \ell c\rp\delta_{\mathrm{F}}}{\lfl T/\ell\rfl}-e^{-\ell\lp c-K\rp} \nonumber\\
&\geq 1-\delta_{\mathrm{F}}-\frac{\ell\exp\lp \ell c\rp\delta_{\mathrm{F}}}{T-\ell}-e^{-\ell\lp c-K\rp} \nonumber\\
&=1-\delta_{\mathrm{F}}-\frac{\delta_{\mathrm{F}}}{sT}\exp\lp\tilde{c}\ell\rp-\varepsilon\label{eq:thm_1_sketch_4}
\end{align}
Rearranging the terms in \eqref{eq:thm_1_sketch_4}, we obtain
\begin{align}
\frac{\delta_{\mathrm{F}}}{sT}\exp\lp\tilde{c}\ell\rp\geq1-\delta_{\mathrm{F}}-\delta_{\mathrm{D}}-\varepsilon.\label{eq:thm_1_sketch_6}
\end{align}
By taking $\log$ on both sides, we have
\begin{align}
\tilde{c}\ell\geq&\log\lp T\rp+\log\lp\frac{1}{\delta_{\mathrm{F}}}\rp+\log\lp s\rp+\log\lp1-\delta_{\mathrm{F}}-\delta_{\mathrm{D}}-\varepsilon\rp. \label{eq:thm_1_sketch_7}
\end{align}
This implies that
\begin{align}
\ell\geq&\frac{1}{\tilde{c}}\lb\log\lp T\rp+\log\lp\frac{1}{\delta_{\mathrm{F}}}\rp+\log\lp s\rp+\log\lp1-\delta_{\mathrm{F}}-\delta_{\mathrm{D}}-\varepsilon\rp\rb.
\end{align}
With the following lemma, we can show that $\tilde{c} \to c$ and $\varepsilon \to 0$ as $T \to \infty$.

\begin{lemma}\label{lem:d_inf}
    For any $\delta_{\mathrm{F}},\delta_{\mathrm{D}}\in\lp0,1\rp$ that satisfy $\delta_{\mathrm{F}}+\delta_{\mathrm{D}}<1$, $\ell\rightarrow\infty$ as $T\rightarrow\infty$
\end{lemma}

\begin{proof}
We prove the lemma by contradiction. First, note that the latency $\ell$ is increasing with $T$ for fixed $\delta_{\mathrm{F}}$ and $\delta_{\mathrm{D}}$. Thus, either the limit of $\ell$ exists or $\ell$ goes to infinity as $T\rightarrow\infty$. Assume that when $\delta_{\mrm{F}} + \delta_{\mrm{D}} < 1$, $\ell$ does not go to infinity as $T\rightarrow\infty$. Then, for any $\bar{\ell}$ larger than the limit of $\ell$, $T\in\mbb{N}$, and $c>C$, by \eqref{eq:thm_1_sketch_4}, we have:
\begin{align}
\delta_{\mathrm{D}}&\geq1-\delta_{\mathrm{F}}-\frac{\bar{\ell}\exp\lp c\bar{\ell}\rp\delta_{\mathrm{F}}}{T-\bar{\ell}}-e^{-\bar{\ell}\lp c-C\rp}.\label{eq:lem_4_proof_1}
\end{align}
Then, by taking $T\rightarrow\infty$, we have:
\begin{align}
\delta_{\mathrm{D}}\geq1-\delta_{\mathrm{F}}-e^{-\bar{\ell}\lp c-C\rp}.\label{eq:lem_4_proof_2}
\end{align}
Last, by taking $\bar{\ell}\rightarrow\infty$, we have $\delta_{\mathrm{D}}\geq1-\delta_{\mathrm{F}}$. Since $\delta_{\mathrm{F}}+\delta_{\mathrm{D}}<1$, this assumption leads to a contradiction.
\end{proof}

By Lemma \ref{lem:d_inf}, it is obvious to see that $\tilde{c} = c + \log\lp\ell\rp/\ell \to c$ and $\varepsilon = e^{-\ell\lp c-K\rp} \to 0$ as $T \to \infty$. 
Additionally, by Theorem \ref{thm:latency_pre-post-known}, $\ell \leq \inf_{\theta\in\lp0,1\rp} \lbp \frac{1}{\lvert\Lambda\lp\theta\rp\rvert} \lb\log\lp\frac{1}{\delta_{\mathrm{D}}}\rp +\theta\beta_{\mrm{C}, r}\lp T, \delta_{\mrm{F}}, r\rp\rb\rbp$ for any $r>1$. Hence, $\ell = \mcal{O}\lp\log\lp T\rp\rp$, implying that $s = \lp T - \ell\rp/\ell \to 1$ as $T \to \infty$. By making $c$ arbitrarily close to $K$ as $T\rightarrow\infty$, \eqref{eq:thm_1_sketch_7} leads to \eqref{eq:d_lower} in the theorem.

\section{Proof of Lemma \ref{lem:GLR-kl-unknown-pre-post}}
\label{sec:lem1}

Recall that $f_{\mu}$ is the Gaussian density with mean $\mu$ and variance $\sigma^{2}$, $\hat{\mu}_{m:n}$ is the empirical mean of the stochastic observations $\lbp X_{m}, \dots, X_{n} \rbp$, and \textup{$ D \lp x;y\rp\coloneqq(x-y)^{2}/(2\sigma^{2})$} is the KL-divergence between two Gaussian distributions with common variance $\sigma^{2}$ and different means $x,y\in\mbb{R}$. We can show that for any $n\in\mbb{N}$ and  $s\in\lb n\rb$,
\begin{align}
&\log\lp\frac{\sup_{\mu_{0}'\in\mbb{R}}\prod_{i=1}^{s}f_{\mu_{0}'}\lp X_{i}\rp\sup_{\mu_{1}'\in\mbb{R}}\prod_{i=s+1}^{n}f_{\mu_{1}'}\lp X_{i}\rp}{\sup_{\mu\in\mbb{R}}\prod_{i=1}^{n}f_{\mu}\lp X_{i}\rp}\rp\nonumber\\
    &=\log\lp\frac{\sup_{\mu_{0}'\in\mbb{R}}\prod_{i=1}^{s}\frac{1}{\sqrt{2\pi\sigma^{2}}}\exp\lp-\frac{\lp X_{i}-\mu_{0}'\rp^{2}}{2\sigma^{2}}\rp\sup_{\mu_{1}'\in\mbb{R}}\prod_{i=s+1}^{n}\frac{1}{\sqrt{2\pi\sigma^{2}}}\exp\lp-\frac{\lp X_{i}-\mu_{1}'\rp^{2}}{2\sigma^{2}}\rp}{\sup_{\mu\in\mbb{R}}\prod_{i=1}^{n}\frac{1}{\sqrt{2\pi\sigma^{2}}}\exp\lp-\frac{\lp X_{i}-\mu\rp^{2}}{2\sigma^{2}}\rp}\rp\nonumber\\
    &=\log\lp\frac{\exp\lp-\inf_{\mu_{0}'\in\mbb{R}}\sum_{i=1}^{s}\frac{\lp X_{i}-\mu_{0}'\rp^{2}}{2\sigma^{2}}\rp\exp\lp-\inf_{\mu_{1}'\in\mbb{R}}\sum_{i=s+1}^{n}\frac{\lp X_{i}-\mu_{1}'\rp^{2}}{2\sigma^{2}}\rp}{\exp\lp-\inf_{\mu\in\mbb{R}}\sum_{i=1}^{n}\frac{\lp X_{i}-\mu\rp^{2}}{2\sigma^{2}}\rp}\rp\nonumber\\
    &\overset{(a)}{=}\log\lp\frac{\exp\lp-\sum_{i=1}^{s}\frac{\lp X_{i}-\hat{\mu}_{1:s}\rp^{2}}{2\sigma^{2}}\rp\exp\lp-\sum_{i=s+1}^{n}\frac{\lp X_{i}-\hat{\mu}_{s+1:n}\rp^{2}}{2\sigma^{2}}\rp}{\exp\lp-\sum_{i=1}^{n}\frac{\lp X_{i}-\hat{\mu}_{1:n}\rp^{2}}{2\sigma^{2}}\rp}\rp\nonumber\\
    &=\log\lp\frac{\exp\lp-\sum_{i=1}^{s}\frac{X_{i}^{2}-2X_{i}\hat{\mu}_{1:s}+\hat{\mu}_{1:s}^{2}}{2\sigma^{2}}\rp\exp\lp-\sum_{i=s+1}^{n}\frac{X_{i}^{2}-2X_{i}\hat{\mu}_{s+1:n}+\hat{\mu}_{s+1:n}^{2}}{2\sigma^{2}}\rp}{\exp\lp-\sum_{i=1}^{n}\frac{X_{i}^{2}-2X_{i}\hat{\mu}_{1:n}+\hat{\mu}_{1:n}^{2}}{2\sigma^{2}}\rp}\rp\nonumber\\
    &=\sum_{i=1}^{s}\frac{2X_{i}\hat{\mu}_{1:s}-\hat{\mu}_{1:s}^{2}}{2\sigma^{2}}+\sum_{i=s+1}^{n}\frac{2X_{i}\hat{\mu}_{s+1:n}-\hat{\mu}_{s+1:n}^{2}}{2\sigma^{2}}-\sum_{i=1}^{n}\frac{2X_{i}\hat{\mu}_{1:n}-\hat{\mu}_{1:n}^{2}}{2\sigma^{2}}\nonumber\\
    &=s\frac{\hat{\mu}_{1:s}^{2}}{2\sigma^{2}}+\lp n-s\rp\frac{\hat{\mu}_{s+1:n}^{2}}{2\sigma^{2}}-n\frac{\hat{\mu}_{1:n}^{2}}{2\sigma^{2}}\nonumber\\
    &=s\frac{\hat{\mu}_{1:s}^{2}}{2\sigma^{2}}+\lp n-s\rp\frac{\hat{\mu}_{s+1:n}^{2}}{2\sigma^{2}}+s\frac{\hat{\mu}_{1:n}^{2}}{2\sigma^{2}}+\lp n-s\rp\frac{\hat{\mu}_{1:n}^{2}}{2\sigma^{2}}-2\frac{s\hat{\mu}_{1:s}\hat{\mu}_{1:n}}{2\sigma^{2}}-2\frac{\lp n-s\rp\hat{\mu}_{s+1:n}\hat{\mu}_{1:n}}{2\sigma^{2}}\nonumber\\
    &=s\frac{\lp\hat{\mu}_{1:s}-\hat{\mu}_{1:n}\rp^{2}}{2\sigma^{2}}+\lp n-s\rp\frac{\lp\hat{\mu}_{s+1:n}-\hat{\mu}_{1:n}\rp^{2}}{2\sigma^{2}}\nonumber\\
    &=s D \lp\hat{\mu}_{1:s};\hat{\mu}_{1:n}\rp+\lp n-s\rp D \lp\hat{\mu}_{s+1:n};\hat{\mu}_{1:n}\rp\label{eq:lem1_proof}
\end{align}
where step $(a)$ follows from the fact that $\sum_{i=t}^{t'}\lp X_{i}-a\rp^{2}$ is minimized when $a=\hat{\mu}_{t:t'}$.

\section{Proof of Theorem \ref{thm:latency_pre-post-unknown}} \label{sec:thm3}

Similar to Theorem \ref{thm:latency_pre-post-known}, there are two change detectors to consider: the GLR and GSR tests. For each test, there are two parts to prove: the false alarm probability $\Pr_{\infty} \lp \tau \leq T \rp$ and the late detection probability $\Pr_{\nu} \lp \tau \geq \nu + \ell \rp$. To prove these probability upper bounds, we first express the statistic using the empirical mean of the observations, so that we can exploit the sub-Gaussianity of the samples and apply concentration inequalities in manner similar to the approach in \citep{besson2022efficient, kaufmann2021mixture}.

We first prove the result for the GLR test in \eqref{eq:GLR_tau-pre-post-unknown}: Recall that $\hat{\mu}_{m:n}$ is the empirical mean of the stochastic observations $\lbp X_{m}, \dots, X_{n} \rbp$ and \textup{$ D \lp x;y\rp\coloneqq(x-y)^{2}/(2\sigma^{2})$} is the KL-divergence between two Gaussian distributions with common variance $\sigma^{2}$ and different means $x,y\in\mbb{R}$. We use the following lemma as our concentration inequality for the false alarm probability upper bound of the GLR test.

\begin{lemma} \label{lem:mix_martin}
Let $\lp X_{n}\rp_{n=1}^{\infty}$ be an i.i.d. $\sigma^{2}$-sub-Gaussian sequence with mean $\mu$, then for all $\delta\in\lp0,1\rp$,
\textup{\begin{align}
    &\Pr_{\infty}\lp\exists\,n\in\mbb{N}:n D \lp\hat{\mu}_{1:n},\mu\rp-3\log\lp1+\log\lp n\rp\rp>\frac{5}{4}\log\lp\frac{1}{\delta}\rp+\frac{11}{2}\rp\leq\delta.\label{eq:mix_mart}
\end{align}}
\end{lemma}
\begin{proof} [Proof of Lemma \ref{lem:mix_martin}]
Let $Y_{n} \coloneqq n D \lp\hat{\mu}_{1:n},\mu\rp-3\log\lp1+\log\lp n\rp\rp$. To prove this lemma, we associate the random process $\{Y_{n}\}$ with a supermartingale, so that we can apply Ville's inequality. To this end, we use the following lemma derived in \citep{kaufmann2021mixture} to construct the associated supermartingale. 
\begin{lemma}[Lemma 13 from \citep{kaufmann2021mixture}]
Let $c\coloneqq\frac{5}{4}\log\lp\frac{\pi^{2}/3}{\lp\log\lp5/4\rp\rp^{2}}\rp$. For any $x > 0$, there exists a nonnegative (mixture) martingale $\lbp Z_{n} \rbp_{n=0}^{\infty}$ such that $Z\lp 0 \rp=1$, and that for any $x>0$ and $n\in\mbb{N}$,
\begin{equation}
    \lbp Y_{n}-c\geq x\rbp\subseteq\lbp Z_{n}\geq e^{\frac{4x}{5}}\rbp\label{eq:mix_martin_include-1}.
\end{equation}
\end{lemma}
Continuing with the proof of Lemma~\ref{lem:mix_martin}, for any $\lambda>0$ and $z>1$, we have:
\begin{align}
    \lbp e^{\lambda\lb Y_{n}-11/2\rb}\geq z\rbp&\overset{(a)}{\subseteq}\lbp e^{\lambda\lb Y_{n}-c\rb}\geq z\rbp \nonumber\\
    &=\lbp Y_{n} - c\geq \frac{\log\lp z\rp}{\lambda}\rbp \nonumber\\
    &\overset{(b)}{\subseteq}\lbp Z_{n}\geq e^{\frac{4\log\lp z\rp}{5\lambda}}\rbp \nonumber\\
    &=\lbp Z_{n}\geq z^{\frac{4}{5\lambda}}\rbp \nonumber\\
    &=\lbp\lp Z_{n}\rp^{5\lambda/4}\geq z\rbp\label{eq:mix_martin_include-2}
\end{align}
where step $(a)$ is owing to the fact that $\frac{11}{2} \geq c$, and step $(b)$ is due to \eqref{eq:mix_martin_include-1}. When $\lambda\leq\frac{4}{5}$, because $g \lp x \rp = x^{5\lambda/4}$ is a concave function, $\lp Z_{n}\rp^{5\lambda/4}$ is a supermartingale.
Hence, for any $\lambda\in\lp0,\frac{4}{5}\rb$, and $y > 11/2$, we have the following inequality:
\begin{align}\nonumber
    \Pr_{\infty}\lp\exists\,n\in\mbb{N}:Y_{n}>y\rp&=\Pr_{\infty}\lp\cup_{n=1}^{\infty}\lbp Y_{n}>y\rbp\rp\\\nonumber
    &=\Pr_{\infty}\lp\cup_{n=1}^{\infty}\lbp e^{\lambda Y_{n}}>e^{\lambda y}\rbp\rp\\\nonumber
    &=\Pr_{\infty}\lp\cup_{n=1}^{\infty}\lbp e^{\lambda\lb Y_{n} - 11/2\rb}>e^{\lambda\lp y - 11/2 \rp}\rbp\rp\\\nonumber
    &\overset{(a)}{\leq}\Pr_{\infty}\lp\cup_{n=1}^{\infty}\lbp\lp Z_{n}\rp^{5\lambda/4}>e^{\lambda \lp y - 11/2 \rp}\rbp\rp\\\nonumber
    &=\Pr_{\infty}\lp\exists\,n\in\mbb{N}:\lp Z_{n}\rp^{5\lambda/4}>e^{\lambda\lp y-11/2\rp}\rp\\\nonumber
    &=\Pr_{\infty}\lp \sup_{n \in \mbb{N}} \lp Z_{n}\rp^{5\lambda/4}>e^{\lambda\lp y-11/2\rp}\rp\\
    &\overset{(b)}{\leq}e^{-\lambda\lp y-11/2 \rp}\label{eq:Ville}
\end{align}
where step $(a)$ is due to \eqref{eq:mix_martin_include-2} and step $(b)$ comes from Ville's inequality \citep{ville1939etude}. By minimizing \eqref{eq:Ville} over $\lambda\in\lp0,\frac{4}{5}\rb$, we obtain
\begin{align}
    \Pr_{\infty}\lp\exists\,n\in\mbb{N}:Y_{n}>y\rp&\leq e^{-\frac{4}{5}\lp y-11/2 \rp}\label{eq:Minimize}.
\end{align}
Then, by letting $\delta=e^{-\frac{4}{5}\lp y-11/2 \rp}$, we can see that for any $\delta\in\lp0,1\rp$,
\begin{align}
    \Pr_{\infty}\lp\exists\,n\in\mbb{N}:Y_{n}>\frac{5}{4}\log\lp\frac{1}{\delta}\rp + \frac{11}{2} \rp&\leq\delta\label{eq:lemma-end}.
\end{align}
\end{proof}

Continuing with the proof of the upper bound on the false alarm probability, for any $T\in\mbb{N}$, we can upper bound the false alarm probability as follows:
\begin{align}\nonumber
    &\Pr_{\infty}\lp \tau_{\mrm{GLR}} \leq T\rp\\\nonumber
    &\leq \Pr_{\infty}\lp \tau_{\mrm{GLR}}<\infty\rp\\\nonumber
    &\overset{(a)}{=}\Pr_{\infty}\lp\exists\,n\in\mbb{N}:\;\sup_{k\in\lb n \rb} k D \lp\hat{\mu}_{1:k},\hat{\mu}_{1:n}\rp+\lp n-k\rp D \lp\hat{\mu}_{k+1:n},\hat{\mu}_{1:n}\rp\geq\beta_{\mrm{GLR}} \lp n,\delta_{\mrm{F}}\rp\rp\\\nonumber
    &=\Pr_{\infty}\left(\exists\,k\leq n\in\mathbb{N}:k D (\hat{\mu}_{1:k},\hat{\mu}_{1:n})+(n-k) D (\hat{\mu}_{k+1:n},\hat{\mu}_{1:n})> \beta_{\mrm{GLR}} \lp n,\delta_{\mrm{F}}\rp\right)\\\nonumber
    &\overset{(b)}{=}\Pr_{\infty}\left(\exists\,k\leq n\in\mbb{N}:\inf_{\mu}k D (\hat{\mu}_{1:k},\mu)+(n-k) D (\hat{\mu}_{k+1:n},\mu)> \beta_{\mrm{GLR}}(n,\delta_{\mrm{F}})\right)\\\nonumber
    &\leq\Pr_{\infty}\left(\exists\,k\leq n\in\mathbb{N}:k D (\hat{\mu}_{1:k},\mu_{0})+(n-k) D (\hat{\mu}_{k+1:n},\mu_{0})> \beta_{\mrm{GLR}}(n,\delta_{\mrm{F}})\right)\\\nonumber
    &=\Pr_{\infty}\Bigg(\exists\,k,r\in\mathbb{N}: s D (\hat{\mu}_{1:k},\mu_{0})+r D (\hat{\mu}_{k+1:k+r},\mu_{0})\\\nonumber
    &\quad\quad\quad\quad >6\log(1+\log(k+r))+\frac{5}{2}\log\lp\frac{4\lp k+r\rp^{\frac{3}{2}}}{\delta_{\mrm{F}}}\rp+11\Bigg)\\\nonumber
    &\overset{(c)}{\leq}\Pr_{\infty}\Bigg(\exists\,k,r\in\mathbb{N}:\lbp k D (\hat{\mu}_{1:k},\mu_{0})>3\log(1+\log(k+r))+\frac{5}{4}\log\lp\frac{4\lp k+r\rp^{3/2}}{\delta_{\mrm{F}}}\rp+\frac{11}{2}\rbp\cup\\\nonumber
    &\quad\quad\quad\Bigg\{ r D (\hat{\mu}_{k+1:k+r},\mu_{0}) > 3\log(1+\log(k+r)) +\frac{5}{4}\log\lp\frac{4\lp k+r\rp^{3/2}}{\delta_{\mrm{F}}}\rp+\frac{11}{2} \Bigg\} \Bigg)\\\nonumber
    &=\Pr_{\infty}\Bigg(\Bigg\{\exists\,k,r\in\mathbb{N}:k D (\hat{\mu}_{1:k},\mu_{0})>3\log(1+\log(k+r))+\frac{5}{4}\log\lp\frac{4\lp k+r\rp^{3/2}}{\delta_{\mrm{F}}}\rp+\frac{11}{2}\Bigg\}\cup\\\nonumber
    &\quad\quad\quad\Bigg\{\exists\,k,r\in\mathbb{N}:r D (\hat{\mu}_{k+1:k+r},\mu_{0})>3\log(1+\log(k+r))+\frac{5}{4}\log\lp\frac{4\lp k+r\rp^{3/2}}{\delta_{\mrm{F}}}\rp+\frac{11}{2}\Bigg\}\Bigg)\\\nonumber
    &\overset{(d)}{\leq}\Pr_{\infty}\bigg(\left\{\exists\,k\in\mathbb{N}:k D (\hat{\mu}_{1:k},\mu_{0})>3\log(1+\log(k))+\frac{5}{4}\log\lp\frac{4}{\delta_{\mrm{F}}}\rp+\frac{11}{2}\right\}\cup\\\nonumber
    &\quad\quad\quad\;\bigg\{\exists\,k,r\in\mathbb{N}:r D (\hat{\mu}_{k+1:k+r},\mu_{0})>3\log(1+\log(r))+\frac{5}{4}\log\lp\frac{4k^{3/2}}{\delta_{\mrm{F}}}\rp+\frac{11}{2}\bigg\}\bigg) \\\nonumber
    &\overset{(e)}{\leq}\Pr_{\infty}\left(\exists\,k\in\mathbb{N}:k D (\hat{\mu}_{1:k},\mu_{0})>3\log(1+\log(k))+\frac{5}{4}\log\lp\frac{4}{\delta_{\mrm{F}}}\rp+\frac{11}{2}\right)+\\\nonumber
    &\quad\;\Pr_{\infty}\left(\exists\,k,r\in\mathbb{N}:r D (\hat{\mu}_{k+1:k+r},\mu_{0})>3\log(1+\log(r))+\frac{5}{4}\log\lp\frac{4k^{3/2}}{\delta_{\mrm{F}}}\rp+\frac{11}{2}\right)\\\nonumber
    &=\Pr_{\infty}\left(\exists\,k\in\mathbb{N}:k D (\hat{\mu}_{1:k},\mu_{0})-3\log(1+\log(k))>\frac{5}{4}\log\lp\frac{4}{\delta_{\mrm{F}}}\rp+\frac{11}{2}\right)+\\\nonumber
    &\quad\;\Pr_{\infty}\bigg(\bigcup_{k=1}^{\infty}\bigg\{\exists\,r\in\mbb{N}:r D (\hat{\mu}_{k+1:k+r},\mu_{0})-3\log(1+\log(r))>\frac{5}{4}\log\lp\frac{4k^{3/2}}{\delta_{\mrm{F}}}\rp + \frac{11}{2}\bigg\}\bigg)\\\nonumber
    &\overset{(f)}{\leq}\Pr_{\infty}\left(\exists\,k\in\mathbb{N}:k D (\hat{\mu}_{1:k},\mu_{0})-3\log(1+\log(k))>\frac{5}{4}\log\lp\frac{4}{\delta_{\mrm{F}}}\rp+\frac{11}{2}\right)+\\\nonumber
    &\quad\;\sum_{k=1}^{\infty}\Pr_{\infty}\bigg(\exists\,r\in\mbb{N}:r D (\hat{\mu}_{k+1:k+r},\mu_{0})-3\log(1+\log(r))>\frac{5}{4}\log\lp\frac{4k^{3/2}}{\delta_{\mrm{F}}}\rp+\frac{11}{2}\bigg)\\\nonumber
    &\overset{(g)}{\leq}\frac{\delta_{\mrm{F}}}{4}+\sum_{k=1}^{\infty}\frac{\delta_{\mrm{F}}}{4k^{3/2}}\\
    &\leq\delta_{\mrm{F}}\label{eq:Prop_2_proof}
\end{align}
where step $(a)$ is due to Lemma \ref{lem:GLR-kl-unknown-pre-post} and step $(b)$ is owing to the fact that $\inf_{\mu} k \lp \hat{\mu}_{1:k} - \mu \rp^{2} + \lp n - k \rp \lp \hat{\mu}_{k + 1 : n} - \mu \rp^{2} = k \lp \hat{\mu}_{1:k} - \hat{\mu}_{1:n} \rp^{2} + \lp n - k \rp \lp \hat{\mu}_{k + 1 : n} - \hat{\mu}_{1:n} \rp^{2}$. Step $(c)$ is due to the fact that $x+y>2a$ implies $x>a$ or $y>a$. Step $(d)$ stems from the fact that $\beta\lp n,\delta\rp$ is increasing with $n$. Steps $(e)$ and $(f)$ are owing to the union bound. By Lemma \ref{lem:mix_martin}, we obtain step $(g)$. This completes the proof of the false alarm constraint in Theorem \ref{thm:latency_pre-post-unknown}.

We now move on to proving the detection delay performance $\Pr_{\nu} \lp \tau_{\mrm{GLR}} \geq \nu + \ell \rp$ in Theorem \ref{thm:latency_pre-post-unknown}. To this end, we use the following lemma borrowed from \citep{besson2022efficient} as our concentration inequality:

\begin{lemma}[Lemma 10 in \citep{besson2022efficient}] \label{lem:sub_Gaussian_diff}
    Let $\hat{\mu}_{i,s}$ be the empirical mean of $s$ i.i.d. $\sigma^{2}$-sub-Gaussian samples with mean $\mu_{i}$, $i\in\lbp0,1\rbp$. Then, $\forall\,s,r\in\mbb{N}$, we have

    \begin{equation}
        \Pr\lp\frac{sr}{s+r}\lp\lp\hat{\mu}_{0,s}-\hat{\mu}_{1,r}\rp-\lp\mu_{0}-\mu_{1}\rp\rp^{2}>u\rp\leq2\exp\lp-\frac{u}{2\sigma^{2}}\rp.
    \end{equation}
\end{lemma}

Continuing with the proof of the latency, for convenience in notation,  using Lemma \ref{lem:GLR-kl-unknown-pre-post}, we can show that for any $T\in\mbb{N},\;\delta_{\mrm{D}},\delta_{\mrm{F}}\in\lp0,1\rp,\;\Delta>0$, $m>\frac{8\sigma^{2}}{\Delta^{2}}\beta_{\mrm{GLR}}\lp T,\delta_{\mrm{F}}\rp$, and $\nu\in\lb m+1,T-\ell \rb$, we have
\begin{align}\nonumber
    &\Pr_{\nu}\lp \tau_{\mrm{GLR}}\geq\nu+\ell\rp\\\nonumber
    &\overset{(a)}{=}\Pr_{\nu}\lp\forall\,n\in\lb \nu+\ell-1\rb:\;\sup_{k\in \lb n \rb} k D \lp\hat{\mu}_{1:k},\hat{\mu}_{1:n}\rp+\lp n-k\rp D \lp\hat{\mu}_{k+1:n},\hat{\mu}_{1:n}\rp<\beta_{\mrm{GLR}}\lp n,\delta_{\mrm{F}}\rp\rp\\\nonumber
    &\overset{(b)}{\leq}\Pr\Bigg(\sup_{k\in \lb\nu+\ell-1\rb} k D \lp\hat{\mu}_{1:k},\hat{\mu}_{1:\nu+\ell-1}\rp+ \lp\nu+\ell-1-k\rp D \lp\hat{\mu}_{k+1:\nu+\ell-1},\hat{\mu}_{1:\nu+\ell-1}\rp\\\nonumber
    &\quad\quad\quad<\beta_{\mrm{GLR}} \lp \nu+\ell-1,\delta_{\mrm{F}}\rp\Bigg)\\\nonumber
    &\leq \Pr_{\nu}\lp\lp\nu-1\rp D \lp\hat{\mu}_{1:\nu-1},\hat{\mu}_{1:\nu+\ell-1}\rp+\ell D \lp\hat{\mu}_{\nu:\nu+\ell-1},\hat{\mu}_{1:\nu+\ell-1}\rp<\beta_{\mrm{GLR}}\lp\nu+\ell-1,\delta_{\mrm{F}}\rp\rp\\\nonumber
    &=\Pr_{\nu}\Bigg(\frac{\nu-1}{2\sigma^{2}}\lp\hat{\mu}_{1:\nu-1}-\frac{\lp\nu-1\rp\hat{\mu}_{1:\nu-1}+\ell\hat{\mu}_{\nu:\nu+\ell-1}}{\nu+\ell-1}\rp^{2}\\\nonumber
    &\quad\quad\quad+\frac{\ell}{2\sigma^{2}}\lp\hat{\mu}_{\nu:\nu+\ell-1}-\frac{\lp\nu-1\rp\hat{\mu}_{1:\nu-1}+\ell\hat{\mu}_{\nu:\nu+\ell-1}}{\nu+\ell-1}\rp^{2}<\beta_{\mrm{GLR}}\lp\nu+\ell-1,\delta_{\mrm{F}}\rp\Bigg)\\
    &=\Pr_{\nu}\lp\frac{\lp\nu-1\rp \ell}{2\sigma^{2}\lp\nu+\ell-1\rp}\lp\hat{\mu}_{1:\nu-1}-\hat{\mu}_{\nu:\nu+\ell-1}\rp^{2}<\beta_{\mrm{GLR}}\lp\nu+\ell-1,\delta_{\mrm{F}}\rp\rp\label{eq:glr-ld-1}
\end{align}
where step $(a)$ comes from Lemma \ref{lem:GLR-kl-unknown-pre-post} and step $(b)$ results from $\lbp\nu+\ell-1\rbp\subseteq\lb \nu+\ell-1\rb$. 

Recall that $\mu_{0}$ and $\mu_{1}$ are the pre- and post-change means, and that the definition of $\ell$ is given in \eqref{eq:d}. For applying Lemma \ref{lem:sub_Gaussian_diff}, we need to convert $\lp\hat{\mu}_{1:\nu-1}-\hat{\mu}_{\nu:\nu+\ell-1}\rp^{2}$ in the last line of \eqref{eq:glr-ld-1} into $\lp\lp\hat{\mu}_{1:\nu-1}-\hat{\mu}_{\nu:\nu+\ell-1}\rp-\lp\mu_{0}-\mu_{1}\rp\rp^{2}$. To this end, we show that for any $\nu\in\lb m+1,T-\ell\rb$, with the choice of $m$ in \eqref{eq:m} and $\ell$ in \eqref{eq:d}, the event  $\lbp\frac{\lp\nu-1\rp \ell}{2\sigma^{2}\lp\nu+\ell-1\rp}\lp\hat{\mu}_{1:\nu-1}-\hat{\mu}_{\nu:\nu+\ell-1}\rp^{2}<\beta_{\mrm{GLR}}\lp\nu+\ell-1,\delta_{\mrm{F}}\rp\rbp$ implies another event $\lbp\frac{\lp\nu-1\rp \ell}{2\sigma^{2}\lp\nu+\ell-1\rp}\lp\lp\hat{\mu}_{1:\nu-1}-\hat{\mu}_{\nu:\nu+\ell-1}\rp-\lp\mu_{0}-\mu_{1}\rp\rp^{2}\geq\beta_{\mrm{GLR}}\lp\nu+\ell-1,\delta_{\mrm{F}}\rp\rbp$. Let $\beta = \beta_{\mrm{GLR}}$ for notational convenience. This implication can be shown as follows:
\begin{align}\nonumber
    &\lbp\frac{\lp\nu-1\rp \ell}{2\sigma^{2}\lp\nu+\ell-1\rp}\lp\hat{\mu}_{1:\nu-1}-\hat{\mu}_{\nu:\nu+\ell-1}\rp^{2}<\beta\lp\nu+\ell-1,\delta_{\mrm{F}}\rp\rbp\\\nonumber
    &\quad\cap\lbp\frac{\lp\nu-1\rp \ell}{2\sigma^{2}\lp\nu+\ell-1\rp}\lp\lp\hat{\mu}_{1:\nu-1}-\hat{\mu}_{\nu:\nu+\ell-1}\rp-\lp\mu_{0}-\mu_{1}\rp\rp^{2}<\beta\lp\nu+\ell-1,\delta_{\mrm{F}}\rp\rbp\\\nonumber
    &=\lbp\lba\hat{\mu}_{1:\nu-1}-\hat{\mu}_{\nu:\nu+\ell-1}\rba<\lp\frac{2\sigma^{2}\lp\nu+\ell-1\rp}{\lp\nu-1\rp \ell}\beta\lp\nu+\ell-1,\delta_{\mrm{F}}\rp\rp^{\frac{1}{2}}\rbp\\\nonumber
    &\quad\cap\lbp\lba\lp\hat{\mu}_{1:\nu-1}-\hat{\mu}_{\nu:\nu+\ell-1}\rp-\lp\mu_{0}-\mu_{1}\rp\rba<\lp\frac{2\sigma^{2}\lp\nu+\ell-1\rp}{\lp\nu-1\rp \ell}\beta\lp\nu+\ell-1,\delta_{\mrm{F}}\rp\rp^{\frac{1}{2}}\rbp\\\nonumber
    &\overset{(a)}{\subseteq}\lbp\lba\hat{\mu}_{1:\nu-1}-\hat{\mu}_{\nu:\nu+\ell-1}\rba<\lp\frac{2\sigma^{2}\lp\nu+\ell-1\rp}{\lp\nu-1\rp \ell}\beta\lp\nu+\ell-1,\delta_{\mrm{F}}\rp\rp^{\frac{1}{2}}\rbp\\\nonumber
    &\quad\cap\lbp\lba\mu_{0}-\mu_{1}\rba-\lba\hat{\mu}_{1:\nu-1}-\hat{\mu}_{\nu:\nu+\ell-1}\rba<\lp\frac{2\sigma^{2}\lp\nu+\ell-1\rp}{\lp\nu-1\rp \ell}\beta\lp\nu+\ell-1,\delta_{\mrm{F}}\rp\rp^{\frac{1}{2}}\rbp\\\nonumber
    &=\lbp\lba\hat{\mu}_{1:\nu-1}-\hat{\mu}_{\nu:\nu+\ell-1}\rba<\lp\frac{2\sigma^{2}\lp\nu+\ell-1\rp}{\lp\nu-1\rp \ell}\beta\lp\nu+\ell-1,\delta_{\mrm{F}}\rp\rp^{\frac{1}{2}}\rbp\\\nonumber
    &\quad\cap\lbp\lba\hat{\mu}_{1:\nu-1}-\hat{\mu}_{\nu:\nu+\ell-1}\rba>\Delta-\lp\frac{2\sigma^{2}\lp\nu+\ell-1\rp}{\lp\nu-1\rp \ell}\beta\lp\nu+\ell-1,\delta_{\mrm{F}}\rp\rp^{\frac{1}{2}}\rbp\\\nonumber
    &\subseteq\lbp\Delta<2\lp\frac{2\sigma^{2}\lp\nu+\ell-1\rp}{\lp\nu-1\rp \ell}\beta\lp\nu+\ell-1,\delta_{\mrm{F}}\rp\rp^{\frac{1}{2}}\rbp\\\nonumber
    &=\lbp\Delta^{2}<8\sigma^{2}\lp\frac{1}{\nu-1}+\frac{1}{\ell}\rp\beta\lp\nu+\ell-1,\delta_{\mrm{F}}\rp\rbp\\\nonumber
    &\overset{(b)}{\subseteq}\lbp\Delta^{2}<8\sigma^{2}\lp\frac{1}{m}+\frac{1}{\ell}\rp\beta\lp T,\delta_{\mrm{F}}\rp\rbp\\\nonumber
    &=\lbp\lp\frac{\Delta^{2}}{8\sigma^{2}\beta\lp T,\delta_{\mrm{F}}\rp}-\frac{1}{m}\rp^{-1}>\ell\rbp\\\nonumber
    &=\lbp \frac{8\sigma^{2}m\beta\lp T,\delta_{\mrm{F}}\rp}{\Delta^{2}m-8\sigma^{2}\beta\lp T,\delta_{\mrm{F}}\rp}>\lce\max\lbp\frac{8\sigma^{2}m\beta\lp T,\delta_{\mrm{F}}\rp}{\Delta^{2}m-8\sigma^{2}\beta\lp T,\delta_{\mrm{F}}\rp},\frac{\delta_{\mrm{F}}^{2/3}}{2^{16/15}\delta_{\mrm{D}}^{4/15}}-m\rbp\rce\rbp\\
    &=\emptyset\label{eq:implication_delay}
\end{align}
where step $(a)$ is due to triangle inequality and step $(b)$ is due to the fact that $\nu\geq m+1$ and $\nu\leq T-\ell$. Hence, the late detection probability can be bounded using Lemma \ref{lem:sub_Gaussian_diff} and we obtain that for any $\nu\in\lb m+1,T-\ell\rb$
\begin{equation}
\begin{aligned}
    &\Pr_{\nu}\lp\tau_{\mrm{GLR}}\geq\nu+\ell\rp\\
    &\leq\Pr_{\nu}\lp\frac{\lp\nu-1\rp \ell}{2\sigma^{2}\lp\nu+\ell-1\rp}\lp\lp\hat{\mu}_{1:\nu-1}-\hat{\mu}_{\nu:\nu+\ell-1}\rp-\lp\mu_{0}-\mu_{1}\rp\rp^{2}\geq\beta_{\mrm{GLR}}\lp\nu+\ell-1,\delta_{\mrm{F}}\rp\rp\\
    &\overset{(a)}{\leq}\Pr_{\nu}\lp\frac{\lp\nu-1\rp \ell}{2\sigma^{2}\lp\nu+\ell-1\rp}\lp\lp\hat{\mu}_{1:\nu-1}-\hat{\mu}_{\nu:\nu+\ell-1}\rp-\lp\mu_{0}-\mu_{1}\rp\rp^{2}\geq\beta_{\mrm{GLR}}\lp m+\ell,\delta_{\mrm{F}}\rp\rp\\
    &\overset{(b)}{\leq}\Pr_{\nu}\lp\frac{\lp\nu-1\rp \ell}{2\sigma^{2}\lp\nu+\ell-1\rp}\lp\lp\hat{\mu}_{1:\nu-1}-\hat{\mu}_{\nu:\nu+\ell-1}\rp-\lp\mu_{0}-\mu_{1}\rp\rp^{2}\geq\frac{5}{2}\log\lp\frac{4\lp m+\ell\rp^{3/2}}{\delta_{\mrm{F}}}\rp\rp\\
    &=\Pr_{\nu}\lp\frac{\lp\nu-1\rp \ell}{\nu+\ell-1}\lp\lp\hat{\mu}_{1:\nu-1}-\hat{\mu}_{\nu:\nu+\ell-1}\rp-\lp\mu_{0}-\mu_{1}\rp\rp^{2}\geq2\sigma^{2}\log\lp\frac{32\lp m+\ell \rp^{15/4}}{\delta_{\mrm{F}}^{5/2}}\rp\rp\\
    &\overset{(c)}{\leq}\frac{\delta_{\mrm{F}}^{5/2}}{16\lp m+\ell\rp^{15/4}}\\
    &\leq\frac{\delta_{\mrm{F}}^{5/2}}{16\lp\delta_{\mrm{F}}^{5/2}2^{-16/15}\delta^{-4/15}_{\mrm{D}}\rp^{15/4}}\\
    &=\delta_{\mrm{D}}
\end{aligned}
\end{equation}
where step $(a)$ is due to the fact that $\beta\lp n,\delta_{\mrm{F}}\rp$ is increasing with $n$, whereas step $(b)$ is owing to the fact that $\beta_{\mrm{GLR}} \lp n,\delta_{\mrm{F}}\rp\geq\frac{5}{2}\log\lp4n^{3/2}/\delta_{\mrm{F}}\rp$. Step $(c)$ comes from Lemma \ref{lem:sub_Gaussian_diff}.

For the GSR test, we first prove the upper bound on $\Pr_{\infty}\lp\tilde{\tau}_{\mrm{GSR}} \leq T \rp$ using \eqref{eq:Prop_2_proof}:
\begin{align}\nonumber
    &\Pr_{\infty}\lp\tilde{\tau}_{\mrm{GSR}} \leq T \rp\\\nonumber
    &=\Pr_{\infty}\lp\exists\,n\in\lb T\rb:\;\log \tilde{W}_{n}\geq\tilde{\beta}_{\mrm{GLR}}\lp n,\delta_{\mrm{F}}\rp + \log n \rp\\\nonumber
    &=\Pr_{\infty}\Bigg(\exists\,n\in\lb T\rb:\\\nonumber
    &\quad\quad\quad\;\sum_{k=1}^{n}\exp\lp k D \lp\hat{\mu}_{1:k},\hat{\mu}_{1:n}\rp+\lp n-k\rp D \lp\hat{\mu}_{k+1:n},\hat{\mu}_{1:n}\rp\rp\geq n\exp\lp\tilde{\beta}_{\mrm{GLR}}\lp n,\delta_{\mrm{F}}\rp\rp\Bigg)\\\nonumber
    &=\Pr_{\infty}\bigg(\exists\,n\in\lb T\rb:\\\nonumber
    &\quad\quad\quad\;\frac{1}{n}\sum_{k=1}^{n}\exp\lp k D \lp\hat{\mu}_{1:k},\hat{\mu}_{1:n}\rp+\lp n-k\rp D \lp\hat{\mu}_{k+1:n},\hat{\mu}_{1:n}\rp\rp\geq\exp\lp\tilde{\beta}_{\mrm{GLR}}\lp n,\delta_{\mrm{F}}\rp\rp\bigg)\\\nonumber
    &\overset{(a)}{\leq}\Pr_{\infty}\bigg(\exists\,n\in\lb T\rb:\\\nonumber
    &\quad\quad\quad\;\sup_{1\leq k\leq n}\exp\lp k D \lp\hat{\mu}_{1:k},\hat{\mu}_{1:n}\rp+\lp n-k\rp D \lp\hat{\mu}_{k+1:n},\hat{\mu}_{1:n}\rp\rp\geq\exp\lp\tilde{\beta}_{\mrm{GLR}}\lp n,\delta_{\mrm{F}}\rp\rp\bigg)\\\nonumber
    &=\Pr_{\infty}\bigg(\exists\,n\in\lb T\rb:\;\sup_{1\leq k\leq n}k D \lp\hat{\mu}_{1:k},\hat{\mu}_{1:n}\rp+\lp n-k\rp D \lp\hat{\mu}_{k+1:n},\hat{\mu}_{1:n}\rp\geq\tilde{\beta}_{\mrm{GLR}}\lp n,\delta_{\mrm{F}}\rp\bigg)\\\nonumber
    &=\Pr_{\infty}\lp\exists\,n\in\lb T\rb:\; \tilde{G}_{n}\geq\tilde{\beta}_{\mrm{GLR}}\lp n,\delta_{\mrm{F}}\rp\rp\\\nonumber
    &=\Pr_{\infty}\lp \tilde{\tau}_{\mrm{GLR}} \leq T \rp\\
    &\overset{(b)}{\leq}\delta_{\mrm{F}}\label{eq:gsr-fa-1}
\end{align}
where step $(a)$ results from the fact that $\sup_{1\leq i\leq n}x_{i}\leq a$ implies $\sum_{i=1}^{n}x_{i}\leq na$, whereas $(b)$ stems from \eqref{eq:Prop_2_proof}. 

Next, we prove the upper bound on the late detection probability $\Pr_{\nu}\lp\tau_{\mrm{GSR}}\geq\nu+\ell\rp$: For any $T\in\mbb{N},\;\,\delta_{\mrm{D}},\delta_{\mrm{F}}\in\lp0,1\rp,\;\Delta>0$, $m>\frac{8\sigma^{2}}{\Delta^{2}}\beta_{\mrm{GSR}}\lp T,\delta_{\mrm{F}}\rp$, and $\nu\in\lb m+1,T-\ell\rb$, we have
\begin{align}\nonumber
    &\Pr_{\nu}\lp\tau_{\mrm{GSR}}\geq\nu+\ell\rp\\\nonumber
    &=\Pr_{\nu}\lp\forall\,n\in \lb \nu+ \ell - 1\rb: \;\log W_{n}<\beta_{\mrm{GSR}}\lp n,\delta_{\mrm{F}}\rp \rp\\\nonumber
    &\overset{(a)}{=}\Pr_{\nu}\Bigg(\forall\,n\in\lb \nu+ \ell-1\rb:\;\log\lp\sum_{k=1}^{n}\exp\lp k D \lp\hat{\mu}_{1:k},\hat{\mu}_{1:n}\rp+\lp n-k\rp D \lp\hat{\mu}_{k+1:n},\hat{\mu}_{1:n}\rp\rp\rp\\\nonumber
    &\quad\quad\quad<\beta_{\mrm{GSR}}\lp n,\delta_{\mrm{F}}\rp\Bigg)\\\nonumber
    &\overset{(b)}{\leq}\Pr_{\nu}\Bigg(\log\lp\sum_{k=1}^{\nu+ \ell - 1} \exp\lp k D \lp\hat{\mu}_{1:k},\hat{\mu}_{1: \nu + \ell - 1} \rp + \lp \nu + \ell - 1 - k\rp D \lp\hat{\mu}_{k+1:\nu+ \ell-1},\hat{\mu}_{1:\nu+ \ell-1}\rp\rp\rp \\\nonumber
    &\quad\quad\quad<\beta_{\mrm{GSR}}\lp \nu + \ell - 1, \delta_{\mrm{F}}\rp\Bigg)\\\nonumber
    &\leq\Pr_{\nu}\lp\lp\nu-1\rp D \lp\hat{\mu}_{1:\nu-1},\hat{\mu}_{1:\nu+ \ell-1}\rp+ \ell D \lp\hat{\mu}_{\nu:\nu+ \ell-1},\hat{\mu}_{1:\nu+ \ell-1}\rp<\beta_{\mrm{GSR}}\lp\nu+ \ell-1,\delta_{\mrm{F}}\rp\rp\\\nonumber
    &=\Pr_{\nu}\Bigg(\frac{\nu-1}{2\sigma^{2}}\lp\hat{\mu}_{1:\nu-1}-\frac{\lp\nu-1\rp\hat{\mu}_{1:\nu-1}+ \ell\hat{\mu}_{\nu:\nu+ \ell-1}}{\nu+ \ell-1}\rp^{2}\\\nonumber
    &\quad\quad\quad+\frac{ \ell}{2\sigma^{2}}\lp\hat{\mu}_{\nu:\nu+ \ell-1}-\frac{\lp\nu-1\rp\hat{\mu}_{1:\nu-1}+ \ell\hat{\mu}_{\nu:\nu+ \ell-1}}{\nu+ \ell-1}\rp^{2}<\beta_{\mrm{GSR}}\lp\nu+ \ell-1,\delta_{\mrm{F}}\rp\Bigg)\\
    &=\Pr_{\nu}\lp\frac{\lp\nu-1\rp \ell}{2\sigma^{2}\lp\nu+ \ell-1\rp}\lp\hat{\mu}_{1:\nu-1}-\hat{\mu}_{\nu:\nu+ \ell-1}\rp^{2}<\beta_{\mrm{GSR}}\lp\nu+ \ell-1,\delta_{\mrm{F}}\rp\rp\label{eq:gsr-ld-1}
\end{align}
where step $(a)$ stems from Lemma \ref{lem:GLR-kl-unknown-pre-post} and step $(b)$ results from $\lbp\nu+ \ell - 1\rbp\subseteq\lb \nu+ \ell - 1 \rb$. Recall that $\mu_{0}$ and $\mu_{1}$ are the pre- and post-change means, and that the definition of $\ell$ is given in \eqref{eq:d}. In order to apply Lemma \ref{lem:sub_Gaussian_diff}, we need to convert $\lp\hat{\mu}_{1:\nu-1}-\hat{\mu}_{\nu:\nu+\ell-1}\rp^{2}$ in the last line of \eqref{eq:gsr-ld-1} into $\lp\lp\hat{\mu}_{1:\nu-1}-\hat{\mu}_{\nu:\nu+\ell-1}\rp-\lp\mu_{0}-\mu_{1}\rp\rp^{2}$. Following the same steps in \eqref{eq:implication_delay} with $\beta = \beta_{\mrm{GSR}}$, we can show that with $m$ and $\ell$ in \eqref{eq:m} and \eqref{eq:d}, the event $\lbp\frac{\lp\nu-1\rp \ell}{2\sigma^{2}\lp\nu+\ell-1\rp}\lp\hat{\mu}_{1:\nu-1}-\hat{\mu}_{\nu:\nu+\ell-1}\rp^{2}<\beta_{\mrm{GSR}}\lp \nu+\ell-1,\delta_{\mrm{F}}\rp\rbp$ implies the event $\lbp\frac{\lp\nu-1\rp \ell}{2\sigma^{2}\lp\nu+\ell-1\rp}\lp\lp\hat{\mu}_{1:\nu-1}-\hat{\mu}_{\nu:\nu+\ell-1}\rp-\lp\mu_{0}-\mu_{1}\rp\rp^{2}\geq\beta_{\mrm{GSR}}\lp \nu+\ell-1,\delta_{\mrm{F}}\rp\rbp$ for any $\nu\in\lb m+1,T-\ell\rbp$.
Then, for any $\nu\in\lb m+1,T-\ell\rb$,
\begin{align}\nonumber
    &\Pr_{\nu}\lp\tau_{\mrm{GSR}}\geq\nu+\ell\rp\\\nonumber
    &\leq\Pr_{\nu}\lp\frac{\lp\nu-1\rp \ell}{2\sigma^{2}\lp\nu\ell-1\rp}\lp\lp\hat{\mu}_{1:\nu-1}-\hat{\mu}_{\nu:\nu+\ell-1}\rp-\lp\mu_{0}-\mu_{1}\rp\rp^{2}\geq\beta_{\mrm{GSR}}\lp\nu+\ell-1,\delta_{\mrm{F}}\rp\rp\\\nonumber
    &\overset{(a)}{\leq}\Pr_{\nu}\lp\frac{\lp\nu-1\rp \ell}{2\sigma^{2}\lp\nu+\ell-1\rp}\lp\lp\hat{\mu}_{1:\nu-1}-\hat{\mu}_{\nu:\nu+\ell-1}\rp-\lp\mu_{0}-\mu_{1}\rp\rp^{2}\geq\beta_{\mrm{GSR}}\lp m+\ell,\delta_{\mrm{F}}\rp\rp\\\nonumber
    &\overset{(b)}{\leq}\Pr_{\nu}\lp\frac{\lp\nu-1\rp \ell}{2\sigma^{2}\lp\nu+\ell-1\rp}\lp\lp\hat{\mu}_{1:\nu-1}-\hat{\mu}_{\nu:\nu+\ell-1}\rp-\lp\mu_{0}-\mu_{1}\rp\rp^{2}\geq\frac{5}{2}\log\lp\frac{4\lp m+\ell\rp^{3/2}}{\delta_{\mrm{F}}}\rp\rp\\\nonumber
    &=\Pr_{\nu}\lp\frac{\lp\nu-1\rp \ell}{\nu+\ell-1}\lp\lp\hat{\mu}_{1:\nu-1}-\hat{\mu}_{\nu:\nu+\ell-1}\rp-\lp\mu_{0}-\mu_{1}\rp\rp^{2}\geq2\sigma^{2}\log\lp\frac{2^{5}\lp m+\ell \rp^{15/4}}{\delta_{\mrm{F}}^{5/2}}\rp\rp\\\nonumber
    &\overset{(c)}{\leq}\frac{\delta_{\mrm{F}}^{5/2}}{16\lp m+\ell\rp^{15/4}}\\\nonumber
    &\leq\frac{\delta_{\mrm{F}}^{5/2}}{16\lp\delta_{\mrm{F}}^{5/2}2^{-16/15}\delta^{-4/15}_{\mrm{D}}\rp^{15/4}}\\
    &=\delta_{\mrm{D}}
\end{align}
where step $(a)$ is due to the fact that $\beta_{\mrm{GSR}}\lp n,\delta_{\mrm{F}}\rp$ is increasing with $n$, whereas step $(b)$ is owing to the fact that $\beta_{\mrm{GSR}}\lp n,\delta_{\mrm{F}}\rp\geq\frac{5}{2}\log\lp4n^{3/2}/\delta_{\mrm{F}}\rp$. Step $(c)$ comes from Lemma \ref{lem:sub_Gaussian_diff}.

\end{document}